\documentclass[12pt]{amsart}
\usepackage{amsmath}
\usepackage{amsthm}
\usepackage{amssymb}
\usepackage{amscd}
\usepackage{bbold}
\usepackage[pdftex]{graphicx}
\usepackage{enumitem}
\usepackage{subcaption}
\usepackage{float}

\usepackage{bbm}
\usepackage[utf8]{inputenc}
\usepackage[english]{babel}
\usepackage{mathrsfs}

\usepackage{appendix}

\newtheorem{lemma}{Lemma}[section]
\newtheorem{theorem}[lemma]{Theorem}
\newtheorem{proposition}[lemma]{Proposition}

\theoremstyle{definition}
\newtheorem{definition}[lemma]{Definition}

\newtheorem{remark}[lemma]{Remark}

\DeclareMathOperator*{\argmin}{arg\,min}

\newcommand{\CC}{{\mathbb C}}

\newcommand{\NN}{{\mathbb N}}

\newcommand{\RR}{{\mathbb R}}

\newcommand{\ZZ}{{\mathbb Z}}

\newcommand{\Ecal}{\mathcal{E}}

\newcommand{\Hcal}{\mathcal{H}}

\newcommand{\Scal}{\mathcal{S}}

\def\benm{\begin{enumerate}}            
\def\eenm{\end{enumerate}}              
\title{Analysis of Decimation on Finite Frames with Sigma-Delta Quantization}

\date{}



\author{Kung-Ching Lin}
\address{Norbert Wiener Center\\
         Department of Mathematics \\
         University of Maryland \\
         College Park, MD 20742 \\
         USA, \\
         Tel. +1 301-405-5058; Fax. +1 301-314-0827}
\email{kclin@math.umd.edu}
\thanks{The author gratefully acknowledges the support of ARO Grant W911NF-17-1-0014}

\begin{document}

\newcounter{bean}

\begin{abstract}
In Analog-to-digital (A/D) conversion, signal decimation has been proven to greatly improve the efficiency of data storage while maintaining high accuracy. When one couples signal decimation with the $\Sigma\Delta$ quantization scheme, the reconstruction error decays exponentially with respect to the bit-rate. In this study, similar results are proven for finite unitarily generated frames. We introduce a process called alternative decimation on finite frames that is compatible with first and second order $\Sigma\Delta$ quantization. In both cases, alternative decimation results in exponential error decay with respect to the bit usage.
\end{abstract}

\maketitle


\section{Introduction}\label{intro}

\subsection{Background and Motivation}\label{background}
	Analog-to-digital (A/D) conversion is a process where bandlimited signals, e.g., audio signals, are digitized for storage and transmission, which is feasible thanks to the classical sampling theorem. In particular, the theorem indicates that discrete sampling is sufficient to capture all features of a given bandlimited signal, provided that the sampling rate is higher than the Nyquist rate.\par
	
	Given a function $f\in L^1(\mathbb{R})$, its Fourier transform $\hat{f}$ is defined as
\[
	\hat{f}(\gamma)=\int_{-\infty}^{\infty} f(t)e^{-2\pi\imath t\gamma}\, dt.
\]
	The Fourier transform can also be uniquely extended to $L^2(\mathbb{R})$ as a unitary transformation.\par
\begin{definition}
	Given $f\in L^2(\mathbb{R})$, $f\in PW_{\Omega}$ if its Fourier transform $\hat{f}\in L^2(\mathbb{R})$ is supported in $[-\Omega,\Omega]$. 
\end{definition}
	An important component of A/D conversion is the following theorem:
\begin{theorem}[Classical Sampling Theorem]
	Given $f\in PW_{1/2}$, for any $g\in L^2(\RR)$ satisfying 
\begin{itemize}
\item	$\hat{g}(\omega)=1$ on $[-1/2, 1/2]$
\item	$\hat{g}(\omega)=0$ for $|\omega|\geq 1/2+\epsilon$,
\end{itemize}
	 with $\epsilon>0$ and $T\in(0,1-2\epsilon)$, $t\in\mathbb{R}$, one has\\
\begin{equation}
	f(t)=T\sum_{n\in\mathbb{Z}}f(nT)g(t-nT),
\end{equation}
	where the convergence is both uniform on compact sets of $\RR$ and in $L^2(\RR)$.
\end{theorem}
	As an extreme case, for $g(t)=\sin(\pi t)/(\pi t)$ and $T=1$, the following identity holds in $L^2(\RR)$:
\[
	f(t)=\sum_{n\in\ZZ} f(n)\frac{\sin(\pi(t-n))}{\pi(t-n)}.
\]
	However, the discrete nature of digital data storage makes it impossible to store exactly the samples $\{f(nT)\}_{n\in\mathbb{Z}}$. Instead, the quantized samples $\{q_n\}_{n\in\mathbb{Z}}$ chosen from a pre-determined finite alphabet $\mathscr{A}$ are stored. This results in the following reconstructed signal
\[
	\tilde{f}(t)=T \sum q_n g(t-nT).
\]
	As for the choice of the quantized samples $\{q_n\}_n$, we shall discuss the following two schemes
\begin{itemize}
\item	Pulse Code Modulation (PCM):\par
	Quantized samples are taken as the direct-roundoff of the current sample, i.e.,
\begin{equation}
\label{Round-off}
	q_n=Q_0(f(nT)):=\argmin_{q\in\mathscr{A}}|q-f(nT)|.
\end{equation}
\item	$\Sigma\Delta$ Quantization:\par
	A sequence of auxiliary variables $\{u_n\}_{n\in\ZZ}$ is introduced for this scheme. $\{q_n\}_{n\in\ZZ}$ is defined recursively as
\[
\begin{split}
	q_n&=Q_0(u_{n-1}+f(nT)),\\
	u_n&=u_{n-1}+f(nT)-q_n.
\end{split}
\]
\end{itemize}
	$\Sigma\Delta$ quantization was introduced \cite{HI_YY_1963} in 1963, and it is still widely used due to some of its advantages over PCM. Specifically, $\Sigma\Delta$ quantization is robust against hardware imperfection \cite{ID_VV_2006}, a decisive weakness for PCM. For $\Sigma\Delta$ quantization, and the more general noise shaping schemes to be explained below, the boundedness of $\{u_n\}_{n\in\ZZ}$ turns out to be essential. Quantization schemes with $\|u\|_\infty<\infty$ are said to be \textit{stable}.\par
	
	Despite its merits over PCM, $\Sigma\Delta$ quantization merely yields linear error decay with respect to the bit-rate as opposed to exponential error decay by its counterpart PCM. Thus, it is desirable to generalize $\Sigma\Delta$ quantization for better error decay rates.
	
	As a direct generalization, given $r\in\NN$, one can consider an \textit{$r$-th order $\Sigma\Delta$ quantization scheme}:
	
\begin{theorem}[Higher Order $\Sigma\Delta$ Quantization, \cite{ID_RD_2003}]
	Given $f\in PW_{1/2}$ and $T<1$, consider the following stable quantization scheme
\[
	f(nT)-q_n=(\Delta^ru):=\sum_{l=0}^r(-1)^l\binom{r}{l}u_{n-l},
\]
	where $\{q_n\}$ and $\{u_n\}$ are the quantized samples and auxiliary variables, respectively. Then, for all $t\in\RR$,
\[
	|f(t)-T\sum_{n\in\ZZ}q_n g(t-nT)|\leq T^r\|u\|_\infty\left\|\frac{d^r g}{dt^r}\right\|_1.
\]
\end{theorem}
	
	Higher order $\Sigma\Delta$ quantization has been known for a long time \cite{WC_PW_RG_1989, PF_AG_RA_1990}, and the $r$-th order $\Sigma\Delta$ quantization improves the error decay rate from linear to polynomial degree $r$ while preserving the advantages of a first order $\Sigma\Delta$ quantization scheme.
	
	From here, a natural question arises: is it possible to generalize $\Sigma\Delta$ quantization further so that the reconstruction error decay matches the exponential decay of PCM? Two solutions have been proposed for this question. The first one is to adopt different quantization schemes. Many of the proposed schemes, including higher order $\Sigma\Delta$ quantization, can be categorized as noise shaping quantization schemes, and a brief summary of such schemes will be provided in Section \ref{prelim}.
	
	The other possibility is to enhance data storage efficiency while maintaining the same level of reconstruction accuracy, and \textit{signal decimation} belongs in this category. Signal decimation is implemented as follows: given an r-th order $\Sigma\Delta$ quantization scheme, there exists $\{q_n^T\}, \{u_n\}$ such that\\
\begin{equation}
\label{quant_def}
	f^{(T)}_n-q_n^T=f(nT)-q_n^T=(\Delta^r u)_n,
\end{equation}
	where $\|u\|_\infty<\infty$, and $\{f^{(T)}_n\}_n=\{f(nT)\}_n$. Then, consider
\[
	\tilde{q}_n^{T_0}:=(S_\rho^r q^T)_{(2\rho+1)n},
\]
	a sub-sampled sequence of $S_\rho^r q^T$, where $(S_\rho h)_n :=\frac{1}{2\rho+1}\sum_{m=-\rho}^{\rho}h_{n+m}$. Signal decimation is the process with which one converts the quantized samples $\{q_n^T\}$ to $\{\tilde{q}_n^{T_0}\}$. See Figure \ref{Diagram} for an illustration.

	Decimation has been known in the engineering community \cite{JC_1986}, and it was observed that decimation results in exponential error decay with respect to the bit-rate, even though the observation remained a conjecture until 2015 \cite{ID_RS_2015}, when Daubechies and Saab proved the following theorem:

\begin{theorem}[Signal Decimation for Bandlimited Functions, \cite{ID_RS_2015}]\label{ID_RS_Deci}
	Given $f\in PW_{1/2}$, $T<1$, and $T_0=(2\rho+1) T< 1$, there exists a function $\tilde{g}$ such that\\
\[
	f(t)=T_0\sum[S_\rho^r f^{(T)}]_{(2\rho+1)n}\tilde{g}(t-nT_0),
\]
\begin{equation}
\label{err_dec_DS}
	|f(t)-T_0\sum \tilde{q}_n^{T_0}\tilde{g}(t-nT_0)|\leq C\|u\|_\infty\left(\frac{T}{T_0}\right)^r =:\mathcal{D},
\end{equation}
	where $\{f_n^{(T)}\}_{n\in\ZZ}$ is defined in \eqref{quant_def}, and $C$ is a constant such that $T_0^r\tilde{g}^{(r)}_1\leq C$. Moreover, the number of bits needed for each unit interval is
\begin{equation}
\label{resource_DS}
	\frac{1}{T_0}\log_{2}((2\rho+1)^r+1)\leq\frac{1}{T_0}\log_{2}\bigg(2\bigg(\frac{T_0}{T}\bigg)^r\bigg)=:\mathcal{R}.
\end{equation}
	Consequently,\\
\[
	\mathcal{D}(\mathcal{R})=2\|u\|_\infty C 2^{-T_0\mathcal{R}}.
\]

\end{theorem}

	From \eqref{err_dec_DS} and \eqref{resource_DS}, we can see that the reconstruction error after decimation still decays polynomially with respect to the sampling rate. As for the data storage, the number of bits needed changes from $O(T^{-1})$ to $O(\log(1/T))$. Thus, the reconstruction error decays exponentially with respect to the bits used.

\begin{figure}
\centering
\includegraphics[scale=1]{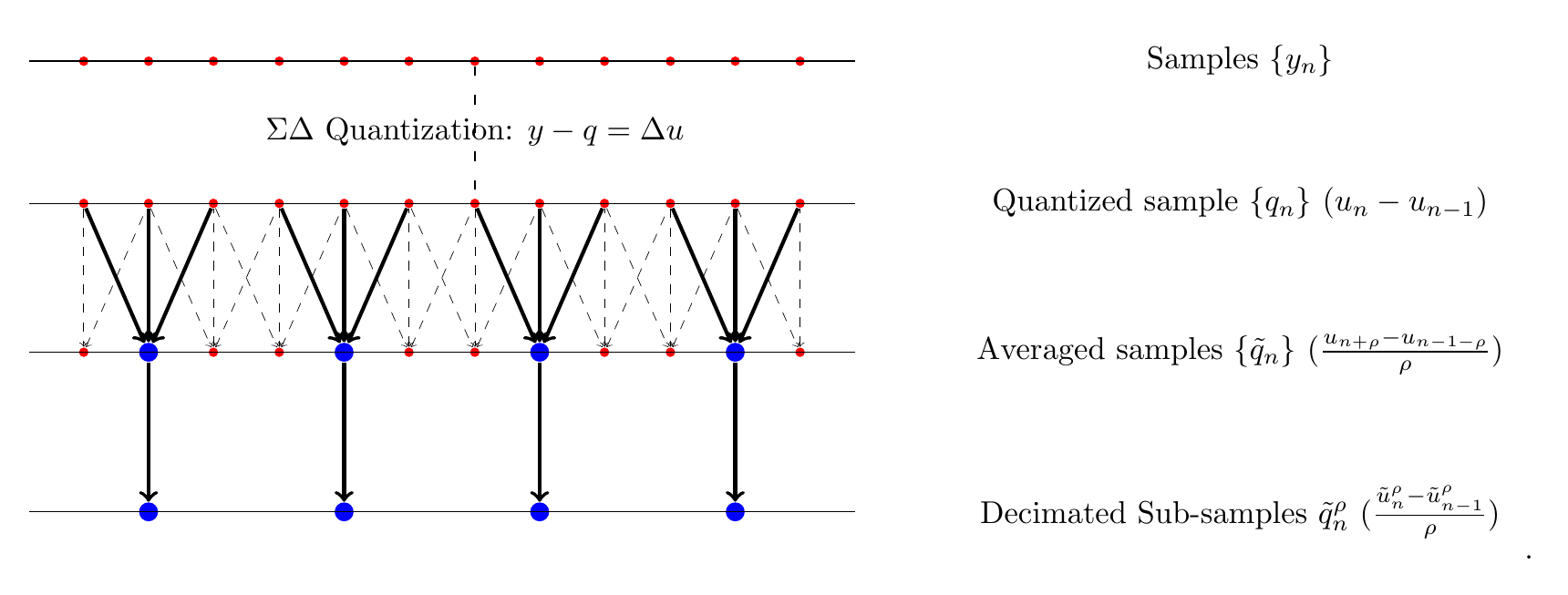}
\caption{Illustration of the first order decimation scheme. After obtaining the quantized samples $\{q_n\}_n$ in the first step, decimation takes the average of quantized samples within disjoint blocks in the second step. The outputs are used as the decimated sub-samples $\{\tilde{q}_n^\rho\}$ in the third step. The effect on the reconstruction (replacing $q_n$ with $y_n-q_n$) is illustrated in parentheses.}
\label{Diagram}
\end{figure}
	
\subsection{Outline and Results}\label{outlines}
	In this paper, we formulate and prove Theorem \ref{main_uni} and Theorem \ref{thm:high_order}, which is an extension of Theorem \ref{ID_RS_Deci} to finite frames. In particular, using our notion of alternative decimation, which will be defined in Section \ref{main_result}, we shall prove exponential error decay with respect to the total number of bits used.
	
	To provide necessary background information, we include preliminaries for signal quantization theory on finite frames in Section \ref{prelim}. We first define $\Sigma\Delta$ quantization on finite frames in Section \ref{SD_quant_frame}. Then, we give a formal definition of noise shaping schemes, which is more general than $\Sigma\Delta$ quantization, in Section \ref{NS_scheme}. Section \ref{prior_work} is devoted to perspective and prior works, and our notation is defined in Section \ref{notation}.
	
	In Section \ref{main_result}, we define alternative decimation and state our main results. Theorem \ref{main} is a special case of Theorem \ref{main_uni}, where we restrict ourselves to \textit{finite harmonic frames}, a subclass of unitarily generated frames. The same result for unitarily generated frames satisfying certain mild conditions is proven in Theorem \ref{main_uni}, and it is further extended to the second order in Theorem \ref{thm:high_order}. The multiplicative structure of decimation is proven in Theorem \ref{mult}, and this enables us to perform decimation iteratively.
	
	We prove Theorems \ref{main}, \ref{main_uni}, \ref{mult}, and \ref{thm:high_order} in Sections \ref{HAR}, \ref{gen_uni}, \ref{mult_discuss}, and \ref{high_order_generalize}, respectively. Generalization to orders greater than two is not possible with our current construction, and we illustrate its main obstacle in Appendix \ref{higher_order_diff}. Decimation for arbitrary orders can be achieved with a different approach and will be introduced in a sequel. Numerical experiments are given in Appendix \ref{num_exp}.

\section{Preliminaries on Finite Frame Quantization}\label{prelim}
	

	Signal quantization theory on finite frames is well motivated from the need to deal with data corruption or erasure \cite{VG_JK_MV_1999, VG_JK_JK_2001}. The authors considered the PCM quantization scheme described above and modeled the quantization error as random noise. In \cite{JB_AP_OY_2006}, deterministic analysis on $\Sigma\Delta$ quantization for finite frames showed that a linear error decay rate is obtained with respect to the oversampling ratio. Moreover, if the frame satisfies certain smoothness conditions, the decay rate can be super-linear for first order $\Sigma\Delta$ quantization. Noise shaping schemes for finite frames have also been investigated, some of which yield exponential error decay rate \cite{CSG_2015,EC_CSG, CSG_2016}. In this section, we shall provide necessary information on quantization for finite frames before stating our results in Section \ref{main_result}.

\subsection{$\Sigma\Delta$ Quantization on Finite Frames}\label{SD_quant_frame}

Fix a separable Hilbert space $\Hcal$ along with a set of vectors $T=\{e_j\}_{j\in\ZZ}\subset\Hcal$. The collection of vectors $T$ forms a frame for $\Hcal$ if there exist $A,B>0$ such that for any $v\in\Hcal$, the following inequality holds:
\[
	A\|v\|_{\Hcal}^2\leq \sum_{j\in\ZZ}|{<}v,e_j{>}|^2\leq B\|v\|_{\Hcal}^2.
\]

	The concept of frames is a generalization of orthonormal bases in a vector space. Different from bases, frames are usually over-complete: the vectors form a linearly dependent spanning set. Over-completeness of frames is particularly useful for noise reduction, and consequently frames are more robust against data corruption than orthonormal bases.\par
	
	Let us restrict ourselves to the case when $\Hcal=\CC^k$ is a finite dimensional Euclidean space, and the frame consists of a finite number of vectors. Given a finite frame $T=\{e_j\}_{j=1}^m$, the linear operator $E:\CC^k\to\CC^m$ satisfying $Ev=\{{<}v,e_j{>}\}_{j=1}^m$ is called the \textit{analysis operator}. Its adjoint operator $E^\ast:\CC^m\to\CC^k$ satisfies $E^\ast c=\sum_{j=1}^mc_je_j$ and is called the \textit{synthesis operator}. The \textit{frame operator} $\Scal$ is defined by $\Scal=E^\ast E: \CC^k\to\CC^k$.

	Under this framework, one considers the quantized samples $q$ of $y=Ex$ and reconstructs $\tilde{x}=\Scal^{-1}E^\ast q$, where $\Scal=E^\ast E$. The frame-theoretic greedy $\Sigma\Delta$ quantization is defined as follows: given a finite alphabet $\mathscr{A}\subset\mathbb{C}$, consider the auxiliary variable $\{u_n\}_{n=0}^{m}$, where we shall set $u_0=0$. For $n=1,\dots,m$, we calculate $\{q_n\}_n$ and $\{u_n\}_n$ as follows:
\begin{equation}
\label{recur}
\begin{split}
		q_n&= Q_0(u_{n-1}+y_n)\\
		u_{n}&=u_{n-1}+y_n-q_n,
\end{split}
\end{equation}
	where $Q_0$ is defined in \eqref{Round-off}. In the matrix form, we have\\
\begin{equation}
\label{sd_quant}
	y-q=\Delta u,
\end{equation}
	where $\Delta\in\mathbb{Z}^{m\times m}$ is the backward difference matrix, i.e., $\Delta_{i,i}=1$ for all $1\leq i\leq m$, and $\Delta_{i,i-1}=-1$ for $2\leq i\leq m$. For an $r$-th order $\Sigma\Delta$ quantization, we have instead $y-q=\Delta^r u$.

	In practice, the quantization alphabet $\mathscr{A}$ is often chosen to be $\mathscr{A}_0$ which is uniformly spaced and symmetric around the origin: given $\delta>0$, we define a mid-rise uniform quantizer $\mathscr{A}_0$ of length $2L$ to be $\mathscr{A}_0=\{(2j+1)\delta/2:\, -L\leq j\leq L-1\}$.

	For complex Euclidean spaces, we define $\mathscr{A}=\mathscr{A}_0+\imath\mathscr{A}_0$. In both cases, $\mathscr{A}$ is called a mid-rise uniform quantizer. Throughout this paper we shall always be using $\mathscr{A}$ as our quantization alphabet.\par

	
\subsection{Noise Shaping Schemes and the Choice of Dual Frames}\label{NS_scheme}

	$\Sigma\Delta$ quantization is a subclass of the more general noise shaping quantization, where the quantization scheme is designed such that the reconstruction error is easily separated from the true signal in the frequency domain. For instance, it is pointed out in \cite{CSG_2016} that the reconstruction error of $\Sigma\Delta$ quantization for bandlimited functions is concentrated in high frequency ranges. Since audio signals have finite bandwidth, it is then possible to separate the signal from the error using low-pass filters. 
	
	Noise shaping quantization has been well established for A/D conversion since the mid 20th century \cite{ST_RH_1978}, and in terms of finite frames, noise shaping schemes generalize the $\Sigma\Delta$ scheme in the following way:
\begin{equation}
	y-q=Hu,
\end{equation}
	where $y, q$, and $u$ are the samples, quantized samples, and the auxiliary variable, respectively, while the transfer matrix $H$ is lower-triangular. Now, given an analysis operator $E$, a transfer matrix $H$, and a dual $F$ to $E$, i.e.\ , $FE=I_k$, the reconstruction error in this setting is 
\[
\|x-Fq\|_2=\|F(Ex-q)\|_2=\|FHu\|_2\leq \|FH\|_{\infty,2}\|u\|_\infty, 
\]
	where $\|\cdot\|_{\infty,2}$ is the operator norm between $\ell^\infty$ and $\ell^2$, i.e.,
\[
	\|T\|_{\infty,2}:=\sup_{\|x\|_\infty=1}\|Tx\|_2.
\]

	The choice of the dual frame $F$ plays a role in the reconstruction error. For instance, \cite{JB_ML_AP_OY_2010} proved that $\argmin_{FE=I_k}\|FH\|_2=(H^{-1}E)^\dagger H^{-1}$, where given any matrix $A$, $A^\dagger$ is defined as the canonical dual $(A^\ast A)^{-1}A^\ast$. More generally, one can consider a $V$-dual, namely $(VE)^\dagger V$, provided that $VE$ is still a frame. With this terminology, decimation can be viewed as a special case of $V$-duals, and conversely every $V$-dual can be associated with corresponding post-processing on the quantized sample $q$.


\subsection{Perspective and Prior Works}\label{prior_work}

\subsubsection{Quantization for Bandlimited Functions}
	Despite its simple form and robustness, $\Sigma\Delta$ quantization only results in linear error decay with respect to the sampling period $T$ as $T\to0$. It was shown \cite{ID_RD_2003,WC_PW_RG_1989, PF_AG_RA_1990} that a generalization of $\Sigma\Delta$ quantization, namely the r-th order $\Sigma\Delta$ quantization, has error decay rate of polynomial order $r$. Leveraging the different constants for this family of quantization schemes, sub-exponential decay can also be achieved. A different family of quantization schemes was proven \cite{CSG_2003} to yield exponential error decay with a small exponent ($c\approx0.07$.) In \cite{PD_CG_FK_2011}, the exponent was improved to $c\approx 0.102$.\par

\subsubsection{Finite Frames}
	$\Sigma\Delta$ quantization can also be applied to finite frames. It was proven \cite{JB_AP_OY_2006} that for any family of finite frames with bounded frame variation, the reconstruction error decays linearly with respect to the oversampling ratio $m/k$, where the corresponding analysis operator $E$ is an $m\times k$ matrix. With different choices of dual frames, \cite{JB_ML_AP_OY_2010} proved that the so-called Sobolev dual achieves minimum induced matrix 2-norm for reconstructions. By carefully matching between the dual frame and the quantization scheme, \cite{CSG_2016} proved that using the $\beta$-dual for random frames results in exponential error decay of near-optimal exponent with high probability.\par
	
\subsubsection{Decimation}
	In \cite{JC_1986}, under the assumption that the noise in $\Sigma\Delta$ quantization is random, it was asserted that decimation greatly reduces the number of bits needed while maintaining the reconstruction accuracy. In \cite{ID_RS_2015}, a rigorous proof was given to show that the assertion is indeed valid, and the reduction of bits used improves the linear error decay into exponential error decay with respect to the bit-rate.
	
\subsubsection{Beta Dual of Distributed Noise Shaping}\label{beta_dual}
	Chou and G{\" u}nturk \cite{CSG_2016, EC_CSG} proposed a distributed noise shaping quantization scheme with beta duals as an example. The definition of a beta dual is as follows:
\begin{definition}[Beta Dual]
	Let $E\in\RR^{m\times k}$ be an analysis operator and $k\mid m$. Recall that $F_V\in\RR^{k\times m}$ is a V-dual of $E$ if\\
\begin{equation}
	F_V=(VE)^{\dagger}V
\end{equation}
	where $V\in\RR^{p\times m}$ such that $VE$ is still a frame.
\end{definition}

	Given $\beta>1$, the $\beta$-dual $F_V=(VE)^{\dagger}V$ has  $V=V_{\beta,m}$, a $k$-by-$m$ block matrix such that each block is $v=[\beta^{-1},\beta^{-2},\dots,\beta^{-m/k}]\in\RR^{1\times m/k}$.
	
	In this case, the transfer matrix $H$ is an $m$-by-$m$ block matrix where each block $h$ is an $m/k$-by-$m/k$ matrix with unit diagonal entries and $-\beta$ as sub-diagonal entries. Under this setting, it is proven that the reconstruction error decays exponentially.
	
	One may notice the similarity between the beta dual and decimation. Indeed, if one chooses $\beta=1$ and normalizes $V$ by $\frac{k}{m}$, the same result as decimation can be obtained, achieving linear error decay with respect to the oversampling ratio and exponential decay with respect to the bit usage. Nonetheless, its generalization to higher order error decay with respect to the oversampling ratio is lacking, whereas the alternative decimation we propose can be extended to the second order. In particular, the raw performance of the second order decimation is superior to the \textit{1-dual} under the same oversampling ratio.

\subsection{Notation}\label{notation}
	The following notation is used in this paper:
\begin{itemize}
\item	$x\in\CC^k$: the signal of interest.
\item $E\in\CC^{m\times k}$: a fixed frame.
\item $y=Ex\in\CC^m$: the sample.
\item	$\rho\in\NN$: the block size of the decimation.
\item	$\eta=\lfloor m/\rho\rfloor\in\NN$: the greatest integer smaller than the ratio $m/\rho$.
\item $\mathscr{A}=\mathscr{A}_0+\imath\mathscr{A}_0\subset\CC$: the quantization alphabet. $\mathscr{A}$ is said to have length $2L$ with gap $\delta$ if $\mathscr{A}_0=\{(2j+1)\delta/2:\, -L\leq j\leq L-1\}$ for some $\delta>0$.
\item	$q\in\CC^m$: the quantized sample obtained from the greedy $\Sigma\Delta$ quantization defined in \eqref{recur}.
\item	$u\in\CC^m$: the auxiliary variable of $\Sigma\Delta$ quantization.
\item $F\in\CC^{k\times m}$: a dual to the analysis operator $E$, i.e.\ $FE=I_k$.
\item	$\mathscr{E}$: the reconstruction error $\mathscr{E}=\|x-Fq\|_2$.
\item	$\mathscr{R}$: total number of bits used to record the quantized sample.
\item	$\Omega\in\CC^{k\times k}$: a Hermitian matrix with eigenvalues $\{\lambda_j\}_{j=1}^k\subset\RR$ and corresponding orthonormal eigenvectors $\{v_j\}_{j=1}^k$.
\item	$\Phi\in\CC^{m\times k}$: the analysis operator of the unitarily generated frame (UGF) with the generator $\Omega$ and the base vector $\phi_0\in\CC^k$.
\end{itemize}

\section{Main Results}\label{main_result}
	
	For the rest of the paper, we shall also assume that our $\Sigma\Delta$ quantization scheme is stable, i.e.\ , $\|u\|_\infty$ remains bounded as the dimension $m\to\infty$. Before we state our results, we shall define the notion of a unitarily generated frame.


\subsection{Unitarily Generated Frames}

	A unitarily generated frame $T_u$ is generated by a cyclic group: given a unit base vector $\phi_0\in\CC^k$ and a Hermitian matrix $\Omega\in\CC^{k\times k}$, the frame elements of $T_u$ are defined as
\begin{equation}
\label{uni_frame}
	\phi_j^{(m)}=U_{j/m}\phi_0,\quad   U_t:=e^{2\pi\imath\Omega t}.
\end{equation}
	The analysis operator $\Phi$ of $T_u$ has $\{\phi_j^\ast\}_j$ as its rows.

	As symmetry occurs naturally in many applications, it is not surprising that unitarily generated frames receive serious attention, and their applications in signal processing abound, \cite{GF_1991, HB_YE_2003, EC_CSG, CSG_2016}. \par
	
	One particular application comes from dynamical sampling, which records the spatiotemporal samples of a signal in interest. Mathematically speaking, one tries to recover a signal $f$ on a domain $D$ from the samples $\{f(X),f_{t_1}(X),\dots, f_{t_N}(X)\}$ where $X\subset D$, and $f_{t_j}=A^{t_j}f$ denotes the evolved signal. Equivalently, one recovers $f$ from $\{{<}A^{t_j}f,e_i{>}\}_{i,j}=\{{<}f,(A^{t_j})^{\ast}e_i{>}\}_{i,j}$, which aligns with the frame reconstruction problems, \cite{AA_JD_IK_2013, AA_JD_IK_2015}. In particular, Lu and Vetterli \cite{YL_MV_2009, YL_MV_2009_2} investigated the reconstruction from spatiotemporal samples for a diffusion process. They noted that one can compensate under-sampled spatial information with sufficiently over-sampled temporal data. Unitarily generated frames represent the cases when the evolution process is unitary and the spatial information is one-dimensional. 
	
	It should be noted that unitarily generated frames are group frames with the generator $G=U_{1/m}$ provided that $U_1=G^m=I_k$, while harmonic frames are tight unitarily generated frames. Here, a frame $T=\{e_j\}_j\subset\Hcal$ is tight if for all $v\in \Hcal$, there exists a constant $A>0$ such that $\sum_{j}|{<}v,e_j{>}|^2=A\|v\|^2$.
	
	 A special class of harmonic frames that we shall discuss is the exponential frame with generator $\Omega$ as a diagonal matrix with integer entries and the base vector $\phi_0=(1,\dots,1)^t/\sqrt{k}$.\par


\subsection{Main Theorems}

	It will be shown that, for unitarily generated frames $\Phi$ satisfying conditions specified in Theorem \ref{main_uni}, $\Sigma\Delta$ quantization coupled with alternative decimation still has linear reconstruction error decay rate with respect to the oversampling ratio $\rho$. As for the data storage, decimation allows for highly efficient storage, and the error decays exponentially with respect to the number of bits used. 

\begin{definition}[Alternative Decimation]
	Given fixed $m,\rho\in\NN$, the $(m,\rho)$-alternative decimation operator is defined to be $D_\rho S_\rho$, where
\begin{itemize}
\item	$ S_\rho=S_\rho^+-S_\rho^-\in\RR^{m\times m}$ is the integration operator satisfying
\begin{equation}
\label{S_definition}
\begin{split}
		&(S_\rho^+)_{l,j}=\left\{\begin{array}{ll}
	\frac{1}{\rho}&\text{if}\quad l\geq\rho,\, l-(\rho-1)\leq j\leq l\\
	0&\text{otherwise},
	\end{array}\right.\\
	&(S_\rho^-)_{l,j}=\left\{\begin{array}{ll}
	\frac{1}{\rho}&\text{if}\quad l\leq\rho-1,\, l+1\leq j\leq m-\rho+l \\
	0&\text{otherwise.}
	\end{array}\right.
\end{split}
\end{equation}
	Here, the cyclic convention is adopted: for any $s\in\ZZ$, $s\equiv s+m$.
\item	$D_\rho\in\NN^{\eta\times m}$ is the sub-sampling operator satisfying
\[
	(D_\rho)_{l,j}=\left\{\begin{array}{ll}
		1&\text{if}\quad j=\rho\cdot l\\
		0&\text{otherwise},
	\end{array}\right.
\]
	where $\eta=\lfloor m/\rho\rfloor$.
\end{itemize}
\end{definition}

\begin{remark}[Canonical Decimation $D_\rho \tilde{S}_\rho$ and Alternative Decimation $D_\rho S_\rho$]
\label{tilde_S}
	It is tempting to consider a closely related circulant matrix $\tilde{S}_\rho$ that satisfies $S_\rho=\tilde{S}_\rho-L$, where $L$ is constant on the first $(\rho-1)$ rows and zero otherwise. Visually, $\tilde{S}_\rho$ and $S_\rho$ has the following form
\begin{equation}
\label{S_shape}
	\tilde{S}_\rho=\frac{1}{\rho}\begin{pmatrix}
	1&0&\hdots&\hdots&0&1&\hdots&1\\
	\vdots&\ddots&\ddots&&&\ddots&\ddots&\vdots\\
	\vdots&&\ddots&0&\hdots&\hdots&0&1\\
	1&\hdots&\hdots&1&&&&\\
	&\ddots&&&\ddots&&&\\
	&&\ddots&&&\ddots&&\\
	&&&\ddots&&&\ddots&\\
	&&&&1&\hdots&\hdots&1
	\end{pmatrix},\,
	S_\rho=\frac{1}{\rho}\begin{pmatrix}
	0&-1&\hdots&\hdots&-1&0&&\\
	&&\ddots&&&\ddots&\ddots&\\
	&&&-1&\hdots&\hdots&-1&0\\
	1&\hdots&\hdots&1&&&&\\
	&\ddots&&&\ddots&&&\\
	&&\ddots&&&\ddots&&\\
	&&&\ddots&&&\ddots&\\
	&&&&1&\hdots&\hdots&1
	\end{pmatrix}.
\end{equation}

	Indeed, $D_\rho\tilde{S}_\rho=D_\rho S_\rho$, so there is no difference between the alternative decimation and canonical decimation. However, we will show in Appendix \ref{alt_dec_nec} that $D_\rho\tilde{S}_\rho^2\neq D_\rho S_\rho^2$, and it is necessary to consider $D_\rho S_\rho^2$ instead of $D_\rho\tilde{S}_\rho^2$ for the second order decimation.

\end{remark}

\begin{definition}[Frame variation]
	Given $A=(A_1,\dots,A_p)\in\CC^{s\times p}$, the frame variation $\sigma(A)$ is defined to be
\[
	\sigma(A)=\sum_{t=1}^{p-1}\|A_t-A_{t+1}\|_2.
\]
\end{definition}

\begin{theorem}[Special Case: Decimation for Harmonic Frames]
\label{main}
	Fix the analysis operator $E=E^{m,k}\in\mathbb{C}^{m\times k}$ with entries $E_{l,j}=\frac{1}{\sqrt{k}}exp(-2\pi\imath (n_j l)/m)$. Suppose $\{n_l\}_{l=1}^k$ are distinct integers in $[-k/2,k/2]$, then the following statements are true:
\begin{itemize}
\item[(a)] {\bf Signal reconstruction:} The matrix $D_\rho S_\rho E\in\CC^{\eta\times k}$ has rank $k$.
\item[(b)]	{\bf Error estimate}: The dual $F=(D_\rho S_\rho E)^\dagger D_\rho S_\rho$ to $E$ has reconstruction error\\
\begin{equation}
	\|x-Fq\|_2\leq\frac{\pi}{2}(\sigma(\bar{F})+\|\bar{F}_\eta\|_2)\|u\|_\infty\frac{1}{\rho},
\end{equation} 
	where $\bar{F}=(\bar{F}_1,\dots,\bar{F}_\eta)$ is the canonical dual of the matrix \\$(\frac{1}{\sqrt{k}}e^{-2\pi\imath\rho ln_j/m})_{l,j}\in\CC^{\eta\times k}$.\par
	Moreover, if $\rho\mid m$, then the reconstruction error $\mathscr{E}$ satisfies\\
\begin{equation}
\label{error_est_har}
	\mathscr{E}:=\|x-Fq\|_2\leq\left\{\begin{array}{ll}
	\frac{\pi^2(k+1)}{\sqrt{3}}\|u\|_\infty\frac{k}{m}&\text{if m, k are even and $n_j$'s are nonzero},\\
	\frac{\pi}{2}(\frac{2\pi(k+1)}{\sqrt{3}}+1)\|u\|_\infty\frac{k}{m}&\text{otherwise}.
	\end{array}\right.
\end{equation}
	In particular, the error decays linearly with respect to the oversampling ratio $m/k$.
\item[(c)] {\bf Efficient data storage}: Suppose the length of the quantization alphabet $\mathscr{A}$ is $2L$, then the decimated samples $D_\rho S_\rho q$ can be encoded by a total of $\mathscr{R}=2\lfloor m/\rho\rfloor\log(2L\rho)=2\eta\log(2L\rho)$ bits. Furthermore, suppose $\eta$ is fixed as $m\to\infty$, then as a function of the total number of bits used, the reconstruction error $\mathscr{E}$ is
\begin{equation}
	\mathscr{E}(\mathscr{R})\leq C_{F,L}\|u\|_\infty 2^{-\frac{1}{2\eta}\mathscr{R}},
\end{equation}
	where $C_{F,L}\leq\pi L(\sigma(\bar{F})+\|\bar{F}_\eta\|)$, and $\bar{F}$ is defined above.\par
	For $\rho\mid m$, we have a better estimate 
\begin{equation}
	\mathscr{E}(\mathscr{R})\leq C_{k,L}\|u\|_\infty 2^{-\frac{1}{2\eta}\mathscr{R}},
\end{equation}
	where $C_{k,L}\leq\frac{\pi kL}{\eta}(\frac{2\pi(k+1)}{\sqrt{3}}+1)$, independent of $\rho$. The optimal exponent $\frac{1}{2k}$ will be achieved in the case $\rho=m/k\in\NN$.
\end{itemize}
\end{theorem}

	The more general result is as follows:
\begin{theorem}[Decimation for Unitarily Generated Frames (UGF)]
\label{main_uni}
	Given $\Omega$, $\phi_0$, $\{\lambda_j\}_j$, $\{v_j\}_j$, and $\Phi=\Phi_{m,k}$ as the generator, base vector, eigenvalues, eigenvectors, and the analysis operator of the corresponding UGF, respectively, suppose
\begin{itemize}
\item	$\{\lambda_j\}_{j=1}^k\subset[-\eta/2,\eta/2]\cap\ZZ$,
\item $C_{\phi_0}=\min_s |{<}\phi_0,v_s{>}|^2>0$, and
\item	$\rho\mid m$,
\end{itemize}
	where $\eta=m/\rho$, then the following statements are true:
\begin{itemize}
\item[(a)] {\bf Signal reconstruction:} $D_\rho S_\rho \Phi_{m,k}\in\CC^{\eta\times k}$ has rank $k$. 
\item[(b)] {\bf Error estimate:} For the dual frame $F=(D_\rho S_\rho \Phi_{m,k})^\dagger D_\rho S_\rho$, the reconstruction error $\mathscr{E}_{m,\rho}$ satisfies
\begin{equation}
\label{error_est_uni}
	\mathscr{E}_{m,\rho}\leq\frac{\pi}{2\eta C_{\phi_0}}(2\pi\max_{1\leq j\leq k}|\lambda_j|+1)\|u\|_\infty\frac{1}{\rho}.
\end{equation}
\item[(c)] {\bf Efficient data storage:} Suppose the length of the quantization alphabet is $2L$, then the total number of bits used to record the quantized samples are $\mathscr{R}=2\eta\log(2L\rho)$ bits. Furthermore, suppose $\eta=m/\rho$ is fixed as $m\to\infty$, then as a function of the total number of bits used, $\mathscr{E}_{m,\rho}$ satisfies
\begin{equation}
	\mathscr{E}(\mathscr{R})\leq C_{k,\phi_0,L,\eta}\|u\|_\infty 2^{-\frac{1}{2\eta}\mathscr{R}},
\end{equation}
	where $C_{k,\phi_0,L,\eta}=\frac{\pi L}{\eta C_{\phi_0}}(2\pi\max_{1\leq j\leq k}|\lambda_j|+1)$, independent of $\rho$.
\end{itemize}
\end{theorem}

\begin{remark}
	For Theorem \ref{main} and \ref{main_uni}, if both the signal and the frame are real, then the total number of bits used will be $\mathscr{R}=\eta\log(2L\rho)$ bits, half the amount needed for the complex case.
\end{remark}

	One additional property of decimation is its multiplicative structure.\par

\begin{theorem}[The Multiplicative Structure of Decimation Schemes]
\label{mult}
	Suppose $\rho\mid m$ and $\rho=\rho_1\rho_2$, then the $(m,\rho)$-decimation is equal to the successive iterations of an $(m,\rho_1)$-decimation coupled by an $(m/\rho_1,\rho_2)$-decimation.
\end{theorem}

	Besides the first order alternative decimation in Theorem \ref{main_uni}, it is also possible to generalize the result to the second order decimation. For such a decimation process, the reconstruction error decays quadratically (as opposed to linearly in Theorem \ref{main_uni}) with respect to the oversampling ratio $\rho$ and exponentially with respect to the bit usage.

\begin{theorem}[Second Order Decimation for UGF]
\label{thm:high_order}
	With the same assumptions as Theorem \ref{main_uni} and the additional requirement that the eigenvalues are nonzero, the following statements are true:
\begin{itemize}
\item[(a)] {\bf Signal reconstruction:} $D_\rho S_\rho^2 \Phi_{m,k}\in\CC^{\eta\times k}$ has rank $k$. 
\item[(b)] {\bf Error estimate:} For the dual frame $F=(D_\rho S_\rho^2 \Phi_{m,k})^\dagger D_\rho S_\rho^2$, the reconstruction error $\mathscr{E}_{m,\rho,r}$ has quadratic error decay rate with respect to the oversampling ratio $\rho$:
\begin{equation}
\label{error_est_uni_high_order}
	\mathscr{E}_{m,\rho,r}\leq\frac{\pi^2}{4\eta C_{\phi_0}}\bigg(9+\eta(2\pi\max_{1\leq j\leq k}|\lambda_j|\frac{1}{\eta})^2\bigg)\|u\|_\infty\frac{1}{\rho^2}.\end{equation}
\item[(c)] {\bf Efficient data storage:} Suppose the length of the quantization alphabet is $2L$, then the total number of bits used to record the quantized samples is $\mathscr{R}=4\eta\log(2Lm)$ bits. Furthermore, suppose $\eta=m/\rho$ is fixed as $m\to\infty$, then as a function of the total number of bits used $\mathscr{E}_{m,\rho}$ satisfies
\begin{equation}
	\mathscr{E}(\mathscr{R})\leq C_{k,\phi_0,L,\eta}\|u\|_\infty 2^{-\frac{1}{2\eta}\mathscr{R}},
\end{equation}
	where $C_{k,\phi_0,L,\eta}=\frac{\pi^2}{4\eta C_{\phi_0}}\bigg(9+\eta(2\pi\max_{1\leq j\leq k}|\lambda_j|\frac{1}{\eta})^2\bigg)(2 L\eta)^2$, independent of $\rho$.
\end{itemize}
\end{theorem}

	To better demonstrate the ideas in the proof, Theorem \ref{main} will be proven separately in Section \ref{HAR} even though it is essentially a special case of Theorem \ref{main_uni}. Theorem \ref{main_uni} will be proven in Section \ref{gen_uni}, and Theorem \ref{mult} in Section \ref{mult_discuss}. The proof of Theorem \ref{thm:high_order} is given in Section \ref{high_order_generalize}.


\section{Decimation for Finite Harmonic Frames}\label{HAR}

	To prove Theorem \ref{main}, we break down the proof into the following steps: first, we investigate properties of $D_\rho S_\rho E$, the decimated version of the frame $E$. Then, we examine the effect of $D_\rho S_\rho\Delta$, which is essential for our error estimate.

\subsection{The Scaling Effect of Decimation}\label{twisting_action}

	Let $E=E^{m,k}=(e_{l,j})_{l,j}=(\frac{1}{\sqrt{k}}\exp(-2\pi\imath n_j l/m))_{l,j}$ where $\{n_j\}_j$ are distinct. For any $\rho\leq m$, we have the following lemma:
	
	\begin{lemma}
\label{twist_har}
	$S_\rho$ and $E$ satisfy
\begin{equation}
	\left\{\begin{array}{lcll}
	S_\rho E=E\bar{C}&\text{if} &n_j\neq0&\forall j\\
	S_\rho E=E\bar{C}-K&\text{if} & n_{j_0}=0& \text{for some }j_0,
	\end{array}\right.
\end{equation}
	where $\bar{C}\in\mathbb{C}^{k\times k}$ is a diagonal matrix with entries
\begin{equation}
\label{C_form}
	\bar{C}_{j,j}=\left\{\begin{array}{lcl}
	\frac{\sin(\rho n_j\pi/m)}{\rho\sin(n_j\pi/m)} e^{\pi\imath(\rho-1) n_j/m}&\text{if}&n_j\neq0\\
	1&\text{if}&n_j=0,
	\end{array}\right.
\end{equation}
	and $K$ is zero except for the $j_0$-th column, where
\[
	K_{l,j_0}=\left\{\begin{array}{ll}
	\frac{m}{\rho\sqrt{k}}&\text{if}\quad 1\leq l\leq \rho-1\\
	0&\text{otherwise}.
	\end{array}\right.
\]
	In either case $D_\rho S_\rho E=D_\rho E\bar{C}$ as $D_\rho K=0$.
\end{lemma}

\begin{remark}
In \eqref{S_shape}, one observes that $S_\rho$ differs from an actual circulant matrix $\tilde{S}_\rho$ by a matrix $L$ with $1/\rho$ on every entry of the first $\rho-1$ rows and zero otherwise. Since $D_\rho L=0$, we can conclude that $D_\rho S_\rho=D_\rho\tilde{S}_\rho$. Thus, it is possible to consider $D_\rho\tilde{S}_\rho$, which is a more natural formulation of decimation than the alternative decimation.
\end{remark}
\begin{proof}
	We start with the computation on $\rho\leq l\leq m$. First, suppose $n_j\neq0$. Then, by \eqref{S_definition},
\[
\begin{split}
	(S_\rho^+ E)_{\rho,j}&=\frac{1}{\rho\sqrt{k}}\sum_{s=1}^{\rho}\exp(-2\pi\imath n_j s/m)\\
				&=\frac{1}{\rho\sqrt{k}}e^{-\pi\imath(\rho+1)n_j/m}\frac{e^{\pi\imath(\rho-1) n_j/m}(1-e^{-2\pi\imath\rho n_j/m)})}{1-\exp(-2\pi\imath n_j/m)}\\
				&=\frac{1}{\sqrt{k}}e^{-\pi\imath(\rho+1)n_j/m}\frac{\sin(\rho n_j\pi/m)}{\rho\sin(n_j\pi/m)}.
\end{split}
\]
For $\rho\leq l\leq m$, 
\begin{equation}
\label{lower_part}
\begin{split}
	(S_\rho E)_{l,j}=(S_\rho^{+}E)_{l,j}&=(S_\rho^{+}E)_{\rho,j}\exp(-2\pi\imath(l-\rho)n_j/m)\\
				&=\frac{1}{\rho\sqrt{k}}\exp(-2\pi\imath ln_j/m)\frac{\sin(\rho n_j\pi/m)}{\sin(n_j\pi/m)}e^{\pi\imath(\rho-1) n_j/m}\\
				&=E_{l,j}\frac{\sin(\rho n_j\pi/m)}{\rho\sin(n_j\pi/m)}e^{\pi\imath(\rho-1) n_j/m} .
\end{split}
\end{equation}
	If $n_j=0$, then $(S_\rho^+ E)_{l,j}=\frac{1}{\sqrt{k}}=E_{l,j}$.\par
	For $l\leq\rho$, we make the following observation:
\begin{equation}
	(S_\rho)_{l,j}+\frac{1}{\rho}=(S_\rho)_{l+\rho,j+\rho},
\end{equation}
	with the cyclic convention on indices. Then for $l\leq\rho-1$, noting that $\exp(-2\pi\imath n_j(s+m)/m)=\exp(-2\pi\imath n_js/m)$,
\begin{equation}
\begin{split}
	(S_\rho E)_{l,j}&=\sum_{s=1}^{m}(S_\rho)_{l,s}E_{s,j}\\
					&=\sum_{s=1}^{m}(S_\rho)_{l+\rho,s+\rho}E_{s,j}-\frac{1}{\rho}\sum_{s=1}^m E_{s,j}\\
					&=\sum_{s=1}^m(S_\rho)_{l+\rho,s+\rho}E_{s+\rho,j}\exp(2\pi\imath\rho n_j/m)-\frac{m}{\rho\sqrt{k}}\delta(n_j)\\
					&=E_{l+\rho,j}e^{2\pi\imath\rho n_j/m}\frac{\sin(\rho n_j\pi/m)}{\rho\sin(n_j\pi/m)}e^{\pi\imath(\rho-1) n_j/m} -\frac{m}{\rho\sqrt{k}}\delta(n_j)\\
					&=E_{l,j}\frac{\sin(\rho n_j\pi/m)}{\rho\sin(n_j\pi/m)}e^{\pi\imath(\rho-1) n_j/m} -\frac{m}{\rho\sqrt{k}}\delta(n_j).
\end{split}
\end{equation}
\end{proof}

	Now we can give the condition for which $D_\rho S_\rho E$ has full rank.
	
\begin{proposition}
	The following statements are equivalent:
\begin{itemize}
\item	$D_\rho S_\rho E$ has full rank.
\item	$\{\rho n_j\}_{j=1}^k$ are distinct residues modulo $m$, and $\rho n_j=0$ modulo $m$ implies $n_j=0$.
\end{itemize}
\end{proposition}
\begin{proof}
	By Lemma \ref{twist_har}, we see that
\begin{equation}
	D_\rho S_\rho E=D_\rho E\bar{C}=D\begin{pmatrix}
	E_1\\
	E_2\\
	\vdots\\
	E_m
	\end{pmatrix}\bar{C}=\begin{pmatrix}
	E_\rho\\
	E_{2\rho}\\
	\vdots\\
	E_{\eta\rho}
	\end{pmatrix}\bar{C}.
\end{equation}
	$D_\rho E$ is a sub-matrix of a Vandermonde matrix with parameters $\{\exp(-2\pi\imath \rho n_j/m)\}_{j=1}^k$. Thus this matrix has full rank if and only if $\{\rho n_j\}_{j=1}^k$ are distinct modulo $m$. On the other hand, $\bar{C}$ is an invertible diagonal matrix if and only if $\bar{C}_{j,j}\neq0$. It is true when $\rho n_j\neq0$ for all $j$ except if $n_j=0$ to begin with.\par
\end{proof}

\begin{remark}\label{dist_nj}
	$|\{-\lfloor \eta/2\rfloor,\dots,\lfloor \eta/2\rfloor\}|\geq \eta$, and if $\{n_j\}_{j=1}^k\subset\{-\lfloor \eta/2\rfloor,\dots,\lfloor \eta/2\rfloor\}$ are distinct residues modulo $m$, then $\{\rho n_j\}_j$ are distinct since elements of $\{-\lfloor \eta/2\rfloor,\dots,\lfloor \eta/2\rfloor\}$ are in different cosets of $(\ZZ/m\ZZ)/ker(\sigma)$ where $\sigma:\mathbb{Z}/m\mathbb{Z}\to\mathbb{Z}/m\mathbb{Z}$ satisfies $\sigma(x)=\rho x$.
\end{remark}

	From Lemma \ref{twist_har}, we see that $ D_\rho S_{\rho} E= D_\rho E\bar{C}$. Thus for any dual $\tilde{F}$ to $D_\rho E\bar{C}$, $\tilde{F}=\bar{C}^{-1}\bar{F}$ where $\bar{F}$ is a dual to $ D_\rho E$. The estimate of $\|\bar{C}^{-1}\|_2$ is described in Proposition \ref{F1_est}, and we need a lemma for this proposition:
\begin{lemma}\label{h_est}
	Given any number $\alpha>1$, the function
\[
	h(x)=h_{\alpha}(x)=\frac{\sin(\alpha x)}{\alpha\sin(x)}
\]
	is even and strictly decreasing in $(0,\pi/(2\alpha))$. Moreover, $\min_{x\in[-\pi/2\alpha,\pi/2\alpha]}h_\alpha(x)\geq\frac{2}{\pi}$.
\end{lemma}
\begin{proof}
	Given any $\alpha>1$, note that $\lim_{x\to0}h(x)=1$. Taking the derivative of $h$, we have\\
\begin{equation}
	h'(x)=\frac{\alpha\cos(\alpha x)\sin(x)-\sin(\alpha x)\cos(x)}{\alpha^2\sin^2(x)}=\frac{\cos(\alpha x)\cos(x)}{\alpha^2\sin^2(x)}\bigg(\alpha\tan(x)-\tan(\alpha x)\bigg).
\end{equation}
	The first factor on the right hand side is even and positive in $(0,\pi/(2\alpha))$, while the second one $\alpha\tan(x)-\tan(\alpha x)$ is odd and decreasing in $(0,\pi/(2\alpha))$ by taking yet another derivative. Thus, the derivative of $h$ is odd and negative in $(0,\pi/(2\alpha)]$. That is, on $I_\alpha=[-\pi/(2\alpha),\pi/(2\alpha)]$, $h$ achieves global maximum at $x=0$ and minimum at $x=\pi/(2\alpha)$. At the minimum point,
\[
	h_\alpha\bigg(\frac{\pi}{2\alpha}\bigg)=\frac{1}{\alpha\sin(\pi/(2\alpha))}\geq \frac{2}{\pi}
\]
	by noting that $\sin(z)\leq z$ for any $z>0$.
\end{proof}

\begin{proposition}
\label{F1_est}
	If $\{n_j\}_{j=1}^k$ are concentrated in $[-\eta/2,\eta/2]$ in $\ZZ/m\ZZ$, then
\[
	\|\bar{C}^{-1}\|_2\leq \frac{\pi}{2}.
\]
\end{proposition}

\begin{proof}
	By \eqref{C_form}, we see that
\[
	|\bar{C}_{l,l}|=\bigg|\frac{1}{\rho}\frac{\sin(\rho n_l\pi/m)}{\sin(n_l\pi/m)}\bigg|
\]
	with the convention that $\sin(\rho\cdot 0)/(\rho\sin(0))=1$. Thus, 
\[
\|\bar{C}^{-1}\|_2=\max_{1\leq l\leq k}\big(|\bar{C}(l)|^{-1}\big)=\bigg(\min_{1\leq l\leq k}\{\bigg|\frac{\sin(\rho n_l\pi/m)}{\rho\sin(n_l\pi/m)}\bigg|\}\bigg)^{-1}.
\]
	Using the result from Lemma \ref{h_est} with $\alpha=\rho\geq1$, we see that $\|\bar{C}^{-1}\|_2\leq \frac{\pi}{2}$.
\end{proof}


\subsection{Effect of $S_\rho$ on the Difference Structure $\Delta$}
	
	Here, we describe the effect of $D_\rho S_\rho\Delta$ in Proposition \ref{difference_structure}, which is directly connected to the proof of Theorem \ref{main}.

\begin{lemma}
	$ S_\rho=\frac{1}{\rho}\bar{\Delta}_\rho\Delta^{-1}$ where 
\begin{equation}
		(\bar{\Delta}_\rho)_{l,j}=\left\{\begin{array}{lcl}
		\delta([j-l])-\delta([\rho+j-l])&\text{if}& j\neq m,\\
		\delta([j-l])&\text{if}&j=m,
		\end{array}\right.
\end{equation}
	and $\delta:\ZZ/m\ZZ\to\{0,1\}$ is the Kronecker delta.
\end{lemma}

\begin{proof}
	Let $\bar{\delta}:\ZZ\to\{0,1\}$ be the Kronecker delta on $\ZZ$. Then,
	
\begin{equation}
\begin{split}
	(\bar{\Delta}_\rho\Delta^{-1})_{l,j}&=\sum_{j\leq t\leq m-1}(\delta([t-l])-\delta([\rho+t-l]))+\delta([j-m])\\
						&=\sum_{j\leq t\leq m}\delta([t-l])-\sum_{j\leq t\leq m-1}(\bar{\delta}(\rho+t-l)+\bar{\delta}(\rho+t-l-m)).
\end{split}
\end{equation}
	By definition, 
\begin{equation}
	\left\{\begin{array}{llcl}
	\sum_{j\leq t\leq m}\delta([t-l])&=1&\text{if}& j\leq l\\
	\sum_{j\leq t\leq m-1}\bar\delta(\rho+t-l)&=1&\text{if}& l\geq \rho+1,\, j\leq l-\rho\\
	\sum_{j\leq t\leq m-1}\bar\delta(\rho+t-l-m)&=1&\text{if}& l\leq \rho-1,\, j\leq m-\rho+l.
	\end{array}\right.
\end{equation}
	Thus, splitting into the cases $l\leq\rho-1$, $l=\rho$, and $l\geq\rho+1$, we see that
\begin{equation}
	\frac{1}{\rho}\bar{\Delta}_\rho\Delta^{-1}=S_\rho,
\end{equation}
	as claimed.
\end{proof}

\begin{proposition}\label{difference_structure}
	Given any $n\in\NN$, let $\Delta^{(n)}\in\NN^{n\times n}$ denote the $n$-dimensional backward difference matrix. For $\rho\vert m$, one has
\[
	D_\rho S_\rho\Delta^{(m)}=\frac{1}{\rho}\Delta^{(m/\rho)}D_\rho.
\]
\end{proposition}
\begin{proof}

	If $\rho\mid m$,
	
\[
	D_\rho S_\rho\Delta^{(m)}=\frac{1}{\rho}D_\rho\bar{\Delta}_\rho.
\]
	Now, note that, for $s\neq m$,
\[
\begin{split}
	(D_\rho\bar{\Delta}_\rho)_{l,s}&=(\bar{\Delta}_\rho)_{l\rho,s}\\
			&=\delta(s-l\rho)-\delta(s+\rho-l\rho)=\delta(s-l\rho)-\delta(s-(l-1)\rho)\\
			&=(\Delta D_\rho)_{l,s}.
\end{split}
\]
	For $s=m$, $(D_\rho\bar{\Delta}_\rho)_{l,m}=\delta(m-l\rho)=(\Delta D_\rho)_{l,m}$.
\end{proof}


\subsection{Proof of Theorem \ref{main}}\label{proof_har}
	Before proving Theorem \ref{main}, we shall need two more lemmas:

\begin{lemma}\label{frame_var_har}
	For any $E=E^{n,k}$ with $n\geq k$, suppose $\{n_j\}_j$ are concentrated between $[-k/2,k/2]$, then $E^\ast$ has frame variation $\sigma(E^\ast)\leq\frac{2\pi(k+1)}{\sqrt{3}}$.
\end{lemma}
\begin{proof}
\[
\begin{split}
	\sigma(E^\ast)&=\frac{1}{\sqrt{k}}\sum_{l=1}^{n-1}(\sum_{j=1}^k|e^{-2\pi\imath ln_j/n}-e^{-2\pi\imath (l+1)n_j/n}|^2)^{1/2}\\
				&=\frac{1}{\sqrt{k}}\sum_{l=1}^{n-1}(\sum_{j=1}^k |1-e^{2\pi\imath n_j/n}|)^{1/2}\\
				&\leq\frac{1}{\sqrt{k}}\sum_{l=1}^{n-1}(\sum_{j=1}^k (2\pi n_j/n)^2)^{1/2}\\
				&\leq\frac{1}{\sqrt{k}}2\pi\frac{n-1}{n}(\sum_{j=-k/2}^{k/2}n_j^2)^{1/2}\\
				&\leq \frac{1}{\sqrt{k}}2\pi\sqrt{2\cdot\frac{k/2(k/2+1)(k+1)}{3}}\\
				&\leq \frac{2\pi(k+1)}{\sqrt{3}}.
\end{split}
\]
\end{proof}

\begin{lemma}[\cite{JB_AP_OY_2006}, Theorem III.7]\label{JB_AP_even}
	Given a stable $\Sigma\Delta$ quantization scheme with a mid-rise uniform quantizer of gap $\delta$, if the frame $T=\{e_j\}_{j=1}^m$ satisfies the zero sum condition
\[
	\sum_{j=1}^m e_j=0,
\]
	then the auxiliary variable $u_m$ has
\[
	|u_m|=\left\{\begin{array}{lcl}
		0,&\text{if}& m \text{ even},\\
		\delta/2, &\text{if}& m \text{ odd}.
	\end{array}\right.
\]
\end{lemma}

	Now we are ready to give the proof of Theorem \ref{main}.\par
\begin{proof} of Theorem \ref{main}:\par
	Adopting the notations above, we see that the reconstruction error is
\begin{equation}
\label{error_est}
\begin{split}
	\|x-\tilde{F}(D_\rho S_\rho q)\|_2&=\|\tilde{F}(D_\rho S_\rho)(y-q)\|_2\\
			&=\|\tilde{F}D_\rho S_\rho\Delta^{(m)} u\|_2\\
			&=\|\frac{1}{\rho}\bar{C}^{-1}\bar{F} \Delta^{(\eta)}D_\rho u\|_2\leq\frac{1}{\rho}\|\bar{C}^{-1}\|_2\|\bar{F} \Delta^{(\eta)}D_\rho u\|_2,
\end{split}
\end{equation}
	where the second equality comes from \eqref{sd_quant}, and the third equality follows from Proposition \ref{difference_structure} along with the fact that $\tilde{F}=\bar{C}^{-1}\bar{F}$ with $\bar{F}$ being the canonical dual frame to $ D_\rho E$.
	Suppose $\bar{F}=(\bar{F}_1, \cdots,\bar{F}_\eta)$, then\\
\begin{equation}
\label{F2_part}
	\bar{F}\Delta^{(\eta)}D_\rho u=\sum_{s=1}^{\eta-1}u_{s\rho}(\bar{F}_s-\bar{F}_{s+1})+u_{\eta\rho}\bar{F}_\eta.
\end{equation}
	By Proposition \ref{F1_est} and \eqref{F2_part},
\begin{equation}
	\frac{1}{\rho}\|\bar{C}^{-1}\|_2\|\bar{F} \Delta^{(\eta)}u^{(\eta)}\|_2\leq\frac{\pi}{2\rho}(\sigma(\bar{F})+\|\bar{F}_\eta\|_2)\|u\|_\infty.
\end{equation}
	For the case $\rho\mid m$, we note that $E^{m/\rho,k}$ is a tight frame with frame bound $\frac{m}{k\rho}$. In particular, $(E^{m/\rho,k})^\ast E^{m/\rho,k}=\frac{m}{k\rho}I_k$. Thus, by Lemma \ref{frame_var_har} ,
\[
	\sigma(\bar{F})\leq  \frac{k}{m/\rho}\frac{2\pi(k+1)}{\sqrt{3}}.
\]
	Thus, we have obtained the following error bound
\begin{equation}
\label{odd_est}
	\mathscr{E}_\rho=\|x-\tilde{F}(D_\rho S_\rho q)\|_2\leq\frac{k}{m/\rho}\frac{\pi}{2\rho}(\frac{2\pi(k+1)}{\sqrt{3}}+1)\|u\|_\infty=\frac{\pi}{2}(\frac{2\pi(k+1)}{\sqrt{3}}+1)\|u\|_\infty\frac{k}{m}.
\end{equation}
	Furthermore, by Lemma \ref{JB_AP_even}, if $m,k$ are even, $n_j$'s are all nonzero, and $\rho\mid m$, then $u_{\eta\rho}=u_m=0$. With that there is a better estimate

\begin{equation}
\label{even_est}
	\mathscr{E}_\rho=\|x-\tilde{F}(D_\rho S_\rho q)\|_2\leq\frac{k}{m/\rho}\frac{\pi}{2\rho}\frac{2\pi(k+1)}{\sqrt{3}}\|u\|_\infty=\frac{\pi^2(k+1)}{\sqrt{3}}\frac{k}{m}\|u\|_\infty.
\end{equation}
	Letting $F=\tilde{F}D_\rho S_\rho$, Theorem \ref{main} (b) is now proven.\par

	For Theorem \ref{main} (c), note that for mid-rise uniform quantizers $\mathscr{A}=\mathscr{A}_0+\imath\mathscr{A}_0$ with length $2L$, each entry $q_j$ of $q\in\CC^m$ is of the form 
\[
q_j=\big((2s_j+1)+\imath(2t_j+1)\big)\frac{\delta}{2},\quad-L\leq s_j,\, t_j\leq L-1.
\] 
	Then, each entry in $D_\rho S_\rho q$ is the average of $\rho$ entries in $q$ which has the form 
\[
	(D_\rho S_\rho q)_j=\big((2\tilde{s}_j+\rho)+\imath(2\tilde{t}_j+\rho)\big)\frac{\delta}{2\rho},\quad-L\rho\leq \tilde{s}_j,\, \tilde{t}_j\leq (L-1)\rho.
\]	
	There are at most $((2L-1)\rho+1)^2\leq (2L\rho)^2$ choices per entry with $\eta=m/\rho$ entries in total. Thus, the vector $D_\rho S_\rho q$ can be encoded by $\mathscr{R}=2\eta\log(2L\rho)$ bits. Noting that $\frac{1}{m}\leq \frac{1}{\eta}\cdot\frac{1}{\rho}$ and
\[
	e^{-\frac{1}{2\eta}\mathscr{R}}=\frac{1}{2L\rho},
\]
	for any estimate we have 
\[
\mathscr{E}\leq C\frac{1}{m}\leq C\frac{1}{\eta}\frac{1}{\rho}= C\frac{2L}{\eta}e^{-\frac{1}{2\eta}\mathscr{R}},
\]	
	for some $C>0$. Substituting the suitable constant for each case, we have
\begin{equation}
	\mathscr{E}(\mathscr{R})\leq C_{F,L}\|u\|_\infty 2^{-\frac{1}{2\eta}\mathscr{R}},
\end{equation}
	where $C_{F,L}\leq\pi L(\sigma(\bar{F})+\|\bar{F}_\eta\|_2)$.
	If $\rho\mid m$, then by \eqref{odd_est}, \eqref{even_est}, 
\begin{equation}
	\mathscr{E}(\mathscr{R})\leq C_{k,L}\|u\|_\infty 2^{-\frac{1}{2\eta}\mathscr{R}},
\end{equation}
	where $C_{k,L}\leq\frac{\pi kL}{\eta}(\frac{2\pi(k+1)}{\sqrt{3}}+1)$, independent of $\rho$.

\end{proof}


\section{Generalization: Decimation on Unitarily Generated Frames}\label{gen_uni}

	Upon examining the proof of Theorem \ref{main}, one can see the following interaction between decimation and the existing sampling scheme:
\begin{itemize}
\item	Commutativity: $D_\rho S_\rho E^{m,k}=E^{m/\rho,k}\bar{C}_{m,\rho}$.
\item	Scalability: $D_\rho S_\rho\Delta^{(m)} u=\frac{1}{\rho}\Delta^{(\eta)}D_\rho u$.	
\end{itemize}
	Fixing the $\Sigma\Delta$ quantization scheme for now, any family of frames satisfying the commutativity condition shall be compatible with decimation, yielding exponential error decay with respect to the bit usage. One example is the unitarily generated frames.\par

	The collection of such elements $T_u=\{\phi_j^{(m)}\}_{j=1}^m$ is the frame of interest.\par
\begin{lemma}\label{twisting_uni}
	For the same $ D_\rho$ and $ S_\rho$ along with the analysis operator $\Phi\in\CC^{m\times k}$ of $T_u$ generated by $(\Omega, \{\lambda_j\}_{j=1}^k, \phi_0)$,
\begin{equation}
\label{uni_comm}
	D_\rho S_\rho\Phi_{m,k}=\begin{pmatrix}
	(\phi_\rho^{(m)})^\ast\\
	(\phi_{2\rho}^{(m)})^\ast\\
	\vdots\\
	(\phi_{\eta\rho}^{(m)})^\ast
	\end{pmatrix}\bar{C}_{m,\rho},
\end{equation}
	where $\eta=\lfloor m/\rho\rfloor$ and $\bar{C}_{m,\rho}=\frac{1}{\rho}\sum_{s=1}^\rho U_{(s-\rho)/m}^\ast$ has eigenvalues $\{e^{\pi\imath(\rho-1)\lambda_j/m}\frac{\sin(\rho\lambda_j\pi/m)}{\rho\sin(\lambda\pi/m)}\}_j$. In particular, if $\rho\mid m$, then
\[
	D_\rho S_\rho\Phi_{m,k}=\Phi_{m/\rho,k}\bar{C}_{m,\rho}.
\]
\end{lemma}

\begin{proof}
	First, note that $S_\rho=\tilde{S}_\rho+L$, where $L$ has value $1/\rho$ on the first $\rho-1$ rows and $0$ otherwise, and $D_\rho L=0$. Moreover, for any $1\leq t\leq m$,
	
\begin{equation}
\label{S_on_Phi}
\begin{split}
	(\tilde{S}_\rho\Phi_{m,k})_t&=(\frac{1}{\rho}\sum_{t-\rho+1}^t U_{s/m}\phi_0)^\ast\\
					&=(\frac{1}{\rho}\sum_{t+1}^{t+\rho}U_{(s-\rho)/m}\phi_0)^\ast\\
					&=\phi_t^\ast\cdot\frac{1}{\rho}\sum_{s=1}^\rho U_{(s-\rho)/m}^\ast=(\Phi_{m,k})_t\cdot\frac{1}{\rho}\sum_{s=1}^\rho U_{(s-\rho)/m}^\ast.
\end{split}
\end{equation}
	Thus, $D_\rho S_\rho\Phi_{m,k}=D_\rho\Phi_{m,k}\bar{C}_{m,\rho}+D_\rho L\Phi_{m,k}=D_\rho\Phi_{m,k}\bar{C}_{m,\rho}$.

	Note that we can diagonalize $U_t=BT_tB^\ast$ where $B$ is a unitary matrix and $T_t$ is a diagonal matrix with entries $\{e^{2\pi\imath\lambda_j t}\}_{j=1}^k$. Then, $B^\ast\bar{C}_{m,\rho}B=\frac{1}{\rho}\sum_{s=1}^\rho(B^\ast U_{(s-\rho)/m} B )^\ast$ is a diagonal matrix, with entries
\begin{equation}
\label{C_est_uni}
\begin{split}
	(B\ast \bar{C}_{m,\rho}B)_{j,j}&=\frac{1}{\rho}\sum_{s=1}^\rho\exp(-2\pi\imath\frac{s-\rho}{m}\lambda_j)\\
			&=\frac{1}{\rho}e^{2\pi\imath\frac{\lambda_j(\rho-1)}{m}}\frac{e^{-2\pi\imath\rho\lambda_j/m}-1}{e^{-2\pi\imath\lambda_j/m}-1}=e^{\pi\imath(\rho-1)\lambda_j/m}\frac{\sin\rho\lambda_j\pi/m}{\rho\sin(\lambda_j\pi/m)}.
\end{split}
\end{equation} 
\end{proof}
	Now, we can find the conditions under which $D_\rho S_\rho\Phi_{m,k}$ has full rank:
\begin{proposition}\label{frame_bound_uni}
	Let $\{v_s\}_{s=1}^k$ be a set of orthonormal eigenvectors of $\Omega$ with eigenvalues $\{\lambda_s\}_{s=1}^{k}$. Suppose
\begin{itemize} 
\item	$\rho\mid m$, 
\item	$\{\rho(\lambda_s-\lambda_t)\}_{s\neq t}$ are nonzero integers modulo $m$, and 
\item	the base vector $\phi_0=\sum_{s}c_sv_s$ satisfies $c_s\neq0$ for all $s$,
\end{itemize} 
	then $\Phi_{m/\rho,k}$ is a frame with frame bounds
\[
\bigg(\frac{m}{\rho}\min_s|c_s|^2\bigg)\|x\|_2^2\leq \sum_{s=1}^{m/\rho}|{<}x,\phi_{s\rho}^{(m)}{>}|^2\leq\bigg(\frac{m}{\rho}\max_s|c_s|^2\bigg)\|x\|_2^2.
\]

	In particular, the frame operator $\Scal_{m/\rho}=\Phi^\ast \Phi$ satisfies $\|\Scal^{-1}_{m/\rho}\|_2\leq\frac{1}{\eta\min|c_i|^2}$.
\end{proposition}
	\begin{proof}
	Suppose the assumptions above are true, then given an arbitrary $x\in\CC^k$,
\[
\begin{split}
\sum_{s=1}^{m/\rho}|{<}x,\phi_{s\rho}^{(m)}{>}|^2&=\sum_{s=1}^{m/\rho}|\sum_{t=1}^k{<}x,v_t{>}{<}v_t,\phi_{s\rho}^{(m)}{>}|^2\\
		&=\sum_{s=1}^{m/\rho}|\sum_{t=1}^k{<}x,v_t{>}{<}U_{-s\rho/m}v_t,\phi_0{>}|^2\\
		&=\sum_{s=1}^{m/\rho}|\sum_{t=1}^ke^{-2\pi\imath s\rho\lambda_t/m}{<}x,v_t{>}{<}v_s,\phi_0{>}|^2\\
		&=\sum_{j,l=1}^k{<}x,v_j{>}\overline{{<}x,v_l{>}}{<}v_j,\rho{>}\overline{{<}v_l,\phi_0{>}}\sum_{s=1}^{m/\rho}e^{-2\pi\imath s\rho(\lambda_j-\lambda_l)/m}\\
		&=\frac{m}{\rho}\sum_{j=1}^k|{<}x,v_j{>}|^2|{<}v_j,\phi_0{>}|^2,
\end{split}
\]
	where the second equality follows from the fact that $U_t$ is unitary, the fourth by expanding the sums, and the last one from the following equality 
\[
\sum_{s=1}^{m/\rho}\exp(-2\pi\imath s\rho(\lambda_j-\lambda_l)/m)=\left\{\begin{array}{lcl}
	\frac{m}{\rho}&\text{if}&j= l\\
	0&\text{if}&j\neq l.
	\end{array}\right.
\]
	Finally, we have
\[
\bigg(\frac{m}{\rho}\min_s|c_s|^2\bigg)\|x\|_2^2\leq \frac{m}{\rho}\sum_{j=1}^k|{<}x,v_j{>}|^2|{<}v_j,\phi_0{>}|^2 \leq\bigg(\frac{m}{\rho}\max_s|c_s|^2\bigg)\|x\|_2^2.
\]

\end{proof}
	Moreover, with the same proof in Proposition \ref{F1_est}, we have the estimate on $\|\bar{C}_{m,\rho}^{-1}\|_2$:
\begin{proposition}
\label{C-1_est_uni}
	If the eigenvalues $\{\lambda_j\}_{j=1}^k$ of the generator $\Omega$ are concentrated between $[-m/(2\rho),m/(2\rho)]$, then
\[
	\|\bar{C}_{m,\rho}^{-1}\|_2=\|B^\ast \bar{C}_{m,\rho}^{-1} B\|_2=\max_{1\leq j\leq k}\Bigg\{\bigg|\frac{\sin(\rho\lambda_j \pi/m)}{\rho\sin(\lambda_j\pi/m)}\bigg|^{-1}\Bigg\}\leq \frac{\pi}{2}.
\]
\end{proposition}

	Also, we need to consider the frame variation of $\Phi_{m/\rho,k}^\ast$.

\begin{lemma}\label{uni_var}
	$\sigma(\Phi_{m/\rho}^\ast)\leq 2\pi\max_{1\leq j\leq k}|\lambda_j|$.
\end{lemma}	

\begin{proof}
	Following the same process of Lemma \ref{frame_var_har}, we see that
\begin{equation}
\label{variation_uni}
\begin{split}
	\sigma(\Phi_{m/\rho}^\ast)&=\sum_{s=1}^{m/\rho-1}\|(U_{s\rho/m}-U_{(s+1)\rho/m})\phi_0\|_2\\
			&=\sum_{s=1}^{m/\rho-1}\|U_{s\rho/m}(1-U_{\rho/m})\phi_0\|_2\\
			&=\sum_{s=1}^{m/\rho-1}\|(1-U_{\rho/m})\phi_0\|_2\\
			&=\sum_{s=1}^{m/\rho-1}\|\sum_{j=1}^k c_j[1-e^{2\pi\imath \lambda_j\rho/m}]v_j\|_2\\
			&=\sum_{s=1}^{m/\rho-1}(\sum_{j=1}^k|c_j|^2|e^{2\pi\imath\lambda_j\rho/m}-1|^{2})^{1/2}\\
			&\leq\sum_{s=1}^{m/\rho-1}\bigg(\sum_{j=1}^k|c_j|^2\cdot\big(2\pi|\lambda_j|\frac{\rho}{m}\big)^{2}\bigg)^{1/2}\\
			&\leq\sum_{s=1}^{m/\rho-1}\bigg(\max_{1\leq j\leq k}2\pi|\lambda_j|\frac{\rho}{m}\bigg)\cdot\|\phi_0\|_2\leq2\pi\max_{1\leq j\leq k}|\lambda_j|\cdot\|\phi_0\|_2.
\end{split}
\end{equation}

\end{proof}

	Now we are ready to prove Theorem \ref{main_uni}.\par
\begin{proof} of Theorem \ref{main_uni}.\par
	First of all, that $D_\rho S_\rho\Phi_{m,k}=\Phi_{m/\rho,k}\bar{C}_{m,\rho}$ has full rank follows from Proposition \ref{frame_bound_uni} and \ref{C-1_est_uni}. For notational clarity, we shall denote $\Phi_{m/\rho,k}=\Phi_{m/\rho}$. \par
	Let $\Scal_{m/\rho}=\Phi_{m/\rho}^\ast\Phi_{m/\rho}$ be the corresponding frame operator, then $\|\Scal^{-1}_{m/\rho}\|_2\leq\rho/(mC_{\phi_0})$ where $C_{\phi_0}:=\min_s|c_s|^2$. Also, note that, by Proposition \ref{difference_structure},
\[
	\Scal_{m/\rho}^{-1}\Phi_{m/\rho}^\ast(D_\rho S_\rho)\Delta^{(m)}u=\frac{1}{\rho}\Scal_{m/\rho}^{-1}(\Phi_{m/\rho}^\ast\Delta^{(m/\rho)})D_\rho u.
\]
	Then, the reconstruction error $\|x-\bar{C}_{m,\rho}^{-1}\Scal_{m/\rho}^{-1}\Phi_{m/\rho}^\ast(D_\rho S_\rho)q\|_2$ is
\begin{equation}
\label{recon_error_uni}
\begin{split}
	\|x-\bar{C}_{m,\rho}^{-1}\Scal_{m/\rho}^{-1}\Phi_{m/\rho}^\ast(D_\rho S_\rho)q\|_2&=\|\bar{C}_{m,\rho}^{-1}\Scal_{m/\rho}^{-1}\Phi_{m/\rho}^\ast(D_\rho S_\rho)\Delta^{(m)}u\|_2\\
	&\leq\frac{1}{\rho}\|\bar{C}_{m,\rho}^{-1}\|_2\|\Scal^{-1}_{m/\rho}\|_2(\sigma(\Phi_{m/\rho}^\ast)+\|\phi_m^{(m)}\|_2)\|D_\rho u\|_\infty\\
			&\leq\frac{\pi}{2\rho}\frac{\rho}{mC_{\phi_0}}(\sigma(\Phi_{m/\rho}^\ast)+\|\phi_m^{(m)}\|_2)\|u\|_\infty,
\end{split}
\end{equation}
	where $\|\bar{C}_{m,\rho}^{-1}\|_2\leq\pi/2$ by Proposition \ref{C-1_est_uni}.\par

	Combining \eqref{recon_error_uni}, Lemma \ref{uni_var}, and the fact that $\|\phi_m^{(m)}\|_2=\|U_{1}\phi_0\|_2=\|\phi_0\|_2=1$, the reconstruction error $\mathscr{E}_{m,\rho}$ can be bounded by
\begin{equation}
\begin{split}
	\mathscr{E}_{m,\rho}&\leq \frac{\pi}{2\rho}\frac{\rho}{mC_{\phi_0}}(\sigma(\Phi_{m/\rho}^\ast)+1)\|u\|_\infty\\
		&\leq\frac{\pi}{2mC_{\phi_0}}(2\pi\max_{1\leq j\leq k}|\lambda_j|+1)\|u\|_\infty\\
		&=\frac{\pi}{2\eta C_{\phi_0}}(2\pi\max_{1\leq j\leq k}|\lambda_j|+1)\|u\|_\infty\frac{1}{\rho}.
\end{split}
\end{equation}
	Theorem \ref{main_uni} (c) follows verbatim from the proof in Theorem \ref{main}.
\end{proof}


\section{The Multiplicative Structure of Decimation Schemes}\label{mult_discuss}
	In this section, we demonstrate the multiplicative structure of alternative decimation. In particular, given $m,\rho,\rho_1,\rho_2\in\NN$ fixed with $\rho=\rho_1\rho_2$ and $\rho\mid m$, consider the following operators:
\[
\left.\begin{array}{ll}
	 D_\rho\in\NN^{(m/\rho)\times m},& S_\rho\in\RR^{m\times m},\\
	D_{\rho_1}\in\mathbb{N}^{(m/\rho_1) \times m}, &S_{\rho_1}\in\mathbb{R}^{m\times m},\\
	D_{\rho_2}\in\mathbb{N}^{( m/\rho)\times ( m/\rho_1)}, &S_{\rho_2}\in\mathbb{R}^{( m/\rho_1)\times (m/\rho_1)}.
\end{array}\right.
\]

	We shall show that $D_\rho S_\rho=D_{\rho_2}S_{\rho_2}D_{\rho_1}S_{\rho_1}$.

\begin{proof} of Theorem \ref{mult}:\par
	The $(m,\rho)$-decimation operator is $D_\rho S_\rho$ while the successive iterations of $(m,\rho_1)$ and $(m/\rho_1,\rho_2)$-decimation combine to be $D_{\rho_2}S_{\rho_2}D_{\rho_1}S_{\rho_1}$.
	
	Note that $D_{\rho_2}D_{\rho_1}=D_\rho$. Then, by Proposition \ref{difference_structure},
\[
\begin{split}
	D_{\rho_2}S_{\rho_2}D_{\rho_1}S_{\rho_1}&=(D_{\rho_2}\bar{\Delta}_{\rho_2})(\Delta^{(m/\rho_1)})^{-1}(D_{\rho_1}\bar{\Delta}_{\rho_1})(\Delta^{(m)})^{-1}\\
			&=\Delta^{(m/\rho_1\rho_2)}D_{\rho_2}(\Delta^{(m/\rho_1)})^{-1}\Delta^{(m/\rho_1)}D_{\rho_1}(\Delta^{(m)})^{-1}\\
			&=\Delta^{(m/\rho)}D_{\rho_2}D_{\rho_1}(\Delta^{(m)})^{-1}\\
			&=D_\rho\bar{\Delta}_\rho(\Delta^{(m)})^{-1}=D_\rho S_\rho,
\end{split}
\]
	which concludes our proof.
	\end{proof}

	The multiplicative property implies the possibility to conduct decimation with multiple steps, gradually down-sizing the dimension $m$. It can be particularly useful for parallel computation and transmission of data through multiple devices with scarce storage resources. In particular, for each stage, it suffices to choose $\rho_j$ to be a small number dividing $m$. It reduces the waiting time between each transmission, and the amplification of quantized sample $q$ will not be large after each stage.\par
	Moreover, although the case where $\rho\nmid m$ does not produce this structure for frames, it is now possible to first reduce $m$ to a number closer to $k$. Only at the last stage do we choose $\rho$ that does not divide $m$. This yields the same result as direct division $m/k$ by the remark above while possibly gaining sharper estimate on the error.

\section{Extension to Second Order Decimation}\label{high_order_generalize}
	So far, we have only defined decimation for the first order $\Sigma\Delta$ quantization, while its counterpart for bandlimited functions, introduced in Section \ref{intro}, applies for arbitrary orders. Due to the boundary effect in finite dimensional spaces, it is harder to extend decimation to arbitrary orders. However, there is no issue generalizing this concept to the second order, as stated in Theorem \ref{thm:high_order}. To prove the theorem, we shall need the following lemmas:

\begin{lemma}[Effect of $D_\rho S_\rho^2$ on the Finite Frame]\label{twisting_uni_high}
	If none of the eigenvalues of $U_{1/m}$ are $1$, then
\begin{equation}
\label{uni_comm_high}
	S_\rho\Phi_{m,k}=\Phi_{m,k}\bar{C}_{m,\rho}.
\end{equation}
	where $\bar{C}_{m,\rho}=\frac{1}{\rho}\sum_{s=1}^\rho U_{(s-\rho)/m}^\ast$ has eigenvalues $\{e^{\pi\imath(\rho-1)\lambda_j/m}\frac{\sin(\rho\lambda_j\pi/m)}{\rho\sin(\lambda\pi/m)}\}_j$. In particular, for any $r\in\NN$,
\[
	D_\rho S_\rho^r\Phi_{m,k}=\Phi_{m/\rho,k}\bar{C}_{m,\rho}^r.
\]
\end{lemma}
\begin{remark}
The proof is very similar to the one of Lemma \ref{twisting_uni}. However, since we are now dealing with $D_\rho S_\rho^r$, we are no longer able to use the fact that $D_\rho L=0$. Instead, we impose the condition that $U_{1/m}$ has no eigenvalue equal to $1$. 
\end{remark}

\begin{proof}
	First, note that if $\mathbbm{1}\in\CC^m$ is the constant vector with value $1$, then
\[
	\mathbbm{1}^\ast\Phi_{m,k}=(\sum_{s=0}^{m-1}U_{s/m}\phi_0)^\ast=\phi_0^\ast B(\sum_{s=0}^{m-1}T_{s/m})^\ast B^\ast=0.
\]
	Given $1\leq t\leq m$, note that $S_\rho=\tilde{S}_\rho-L$, where $L$ has value $1/\rho$ on the first $\rho-1$ rows and $0$ otherwise, and $L\Phi_{m,k}=0$. Then, by \eqref{S_on_Phi}, $S_\rho\Phi_{m,k}=\tilde{S}_\rho\Phi_{m,k}=\Phi_{m,k}\bar{C}_{m,\rho}$. Using induction on $r$, $S_\rho^r\Phi_{m,k}=\Phi_{m,k}\bar{C}_{m,\rho}^r$, and $D_\rho S_\rho^r \Phi_{m,k}=\Phi_{m/\rho,k} \bar{C}_{m,\rho}^r$. The properties of $\bar{C}_{m,\rho}$ follow from Lemma \ref{twisting_uni}.
\end{proof}

\begin{lemma}\label{high_commute}
	For any $r,m,\rho\in\NN$,
\[
	D_\rho\bar{\Delta}_\rho^r=(\Delta^{(m/\rho)})^r D_\rho.
\]
\end{lemma}
\begin{proof}
	By Proposition \ref{difference_structure},
\[
	D_\rho\bar{\Delta}_\rho=D_\rho (\bar{\Delta}_\rho (\Delta^{(m)})^{-1}) \Delta^{(m)}=\Delta^{(m/\rho)}D_\rho.
\]
	Thus, for $r\in\NN$, we have, by induction on $r$,
\[
	D_\rho\bar{\Delta}_\rho^r=\Delta^{(m/\rho)}D_\rho\bar{\Delta}_\rho^{r-1}=(\Delta^{(m/\rho)})^r D_\rho.
\]
\end{proof}

\begin{lemma}
\label{non-commute}
	$\Delta^{-1}\bar{\Delta}_\rho\Delta=\bar{\Delta}_\rho+\Ecal$, where $\Ecal_{l,s}=\delta(s-(m-\rho))$.  
\end{lemma}
\begin{proof}
	For $s\neq m$,
\[
\begin{split}
	(\Delta^{-1}\bar{\Delta}_\rho\Delta)_{l,s}&=\sum_{j,n}\Delta^{-1}_{l,j}(\bar{\Delta}_\rho)_{j,n}\Delta_{n,s}\\
				&=\sum_{j=1}^l(\bar{\Delta}_\rho)_{j,s}-(\bar{\Delta}_\rho)_{j,s+1}\\
				&=\sum_{j=1}^l\bigg[\delta(s-j)-\delta(s+\rho-j)-\delta(s+1-j)+\delta(s+1+\rho-j)\bigg]\\
				&=\sum_{j=1}^l(\delta(s-j)-\delta(s+1-j))-\sum_{j=1}^l(\delta(s+\rho-j)-\delta(s+1+\rho-j))\\
				&=\delta(s-l)-\delta(s+\rho-l)+\delta(s+\rho)=(\bar{\Delta}_\rho)_{l,s}+\delta(s+\rho),
\end{split}
\]
	where the $\delta(s+\rho)=\delta(s-(m-\rho))$ comes from the second term in the second-to-last line. When $s+1+\rho=m+1$, the term $\delta(s+1+\rho-j)$ wraps around, producing an additional $-1$.
	
	When $s=m$,
\[
	(\Delta^{-1}\bar{\Delta}_\rho\Delta)_{l,s}=\sum_{j}\Delta^{-1}_{l,j}(\bar{\Delta}_\rho)_{j,m}=\sum_{j=1}^l\delta(m-j)=\delta(m-l).
\]
	Combining the two equations above, we see that $\Delta^{-1}\bar{\Delta}_\rho\Delta=\bar{\Delta}_\rho+\Ecal$.
\end{proof}

\begin{proposition}\label{small_error}
	For $\Phi_{m/\rho,k}^\ast=(\phi_1^{(\eta)}\mid\dots\mid \phi_\eta^{(\eta)})$,
\[
	\Phi_{m/\rho,k}^\ast D_\rho S_\rho^2\Delta^2=\frac{1}{\rho^2}\Phi_{m/\rho,k}^\ast\Delta^2 D_\rho +\frac{1}{\rho^2}V,
\]
	where $V$ is zero except for the $(m-\rho)$-th column, which is $\phi_1^{(\eta)}$.

\end{proposition}

\begin{proof}
	When $r=2$, we consider the $(2,\rho)$-decimation operator $D_\rho S_\rho^2$. Then,
\[
\begin{split}
	D_\rho S_\rho^2\Delta^2&=\frac{1}{\rho^2}D_\rho\bar{\Delta}_\rho\Delta^{-1}\bar{\Delta}_\rho\Delta\\
				&=\frac{1}{\rho^2}D_\rho\bar{\Delta}_\rho(\bar{\Delta}_\rho+\Ecal)\\
				&=\frac{1}{\rho^2}\Delta^2 D_\rho+\frac{1}{\rho^2}D_\rho\bar{\Delta}_\rho\Ecal,
\end{split}
\]
	where the first term in the last line follows from Lemma \ref{high_commute}. Now, $(\bar{\Delta}_\rho\Ecal)_{l,s}=\delta(l-\rho)\delta(s+\rho)$, and $(D_\rho\bar{\Delta}_\rho\Ecal)_{l,s}=\delta(l-1)\delta(s+\rho)$. Thus,
\[
	\Phi_{m/\rho,k}^\ast D_\rho S_\rho^2\Delta^2=\frac{1}{\rho^2}\Phi_{m/\rho,k}^\ast\Delta^2 D_\rho+\frac{1}{\rho^2}V.
\]

\end{proof}

\begin{lemma}\label{uni_var_higher}
	For any $r\in\NN$, $\sum_{s=1}^{m/\rho}\|\Phi_{m/\rho}^\ast\Delta^r v_s\|_2\leq r2^r+\eta(2\pi\max_{1\leq j\leq k}|\lambda_j|\frac{1}{\eta})^r$, where $(v_s)_{j}=\delta(s-j)$, the $s$-th canonical coordinate.
\end{lemma}	

\begin{proof}
	Note that for any $A=(A_1,\dots, A_n)$, $A\Delta e_s=A_s-A_{s+1}$. Thus,
\begin{equation}
\label{variation_uni}
\begin{split}
	\sum_{s<m/\rho-r}\|\Phi_{m/\rho}^\ast\Delta^r v_s\|_2&=\sum_{s<m/\rho-r}\|\sum_{t=s}^{s+r} (-1)^t\binom{r}{t-s}U_{t\rho/m}\phi_0\|_2\\
			&=\sum_{s<m/\rho-r}\|U_{s\rho/m}\sum_{t=0}^r(-1)^t\binom{r}{t}U_{t\rho/m}\phi_0\|_2\\
			&=\sum_{s<m/\rho-r}\|\sum_{t=0}^r(-1)^t\binom{r}{t}U_{t\rho/m}\phi_0\|_2\\
			&=\sum_{s<m/\rho-r}\|\sum_{j=1}^k c_j[\sum_{t=0}^r(-1)^t\binom{r}{t}e^{2\pi\imath t\lambda_j\rho/m}]v_j\|_2\\
			&=\sum_{s<m/\rho-r}(\sum_{j=1}^k|c_j|^2|e^{2\pi\imath\lambda_j\rho/m}-1|^{2r})^{1/2}\\
			&\leq\sum_{s<m/\rho-r}\bigg(\sum_{j=1}^k|c_j|^2\cdot\big(2\pi|\lambda_j|\frac{\rho}{m}\big)^{2r}\bigg)^{1/2}\\
			&\leq\sum_{s<m/\rho-r}\bigg(\max_{1\leq j\leq k}2\pi|\lambda_j|\frac{1}{\eta}\bigg)^r\cdot\|\phi_0\|_2\\
			&\leq\eta(2\pi\max_{1\leq j\leq k}|\lambda_j|\frac{1}{\eta})^r\cdot\|\phi_0\|_2,
\end{split}
\end{equation}
	where we note that $m/\rho=\eta$.
	
	For $s\geq m/\rho-r$, with trivial estimates one has 
\[
	\sum_{s\geq m/\rho-r}\|\Phi_{m/\rho}^\ast\Delta^r e_s\|_2\leq r\|\sum_{j=1}^k |c_j|\sum_{t=0}^r\binom{r}{t}v_j\|\leq r2^r.
\]

\end{proof}

\begin{proposition}["Frame Variation" Estimate]\label{fra_var_est_high}
\[
	\|\Phi_{m/\rho,k}^\ast D_\rho S_\rho^2 \Delta^2 u\|_2\leq \left(9+\eta\left(2\pi\max|\lambda_j|\frac{1}{\eta}\right)^2\right)\frac{1}{\rho^2}.
\]
\end{proposition}

\begin{proof}
	Let $\{v_s\}_s$ be the canonical basis of $\CC^{\eta}$. Then, by Proposition \ref{small_error}, we have
\[
\begin{split}
	\|\Phi_{m/\rho,k}^\ast D_\rho S_\rho^2 \Delta^2 u\|_2&=\frac{1}{\rho^2}\|\Phi_{m/\rho,k}^\ast\Delta^2 D_\rho u+Vu\|_2\\
					&\leq\frac{1}{\rho^2}\left(\sum_{s=1}^{m/\rho}\|\Phi_{m/\rho}^\ast\Delta^2 v_s\|_2+\|\phi_1\|_2\right)\|u\|_\infty\\
					&\leq\frac{1}{\rho^2}\left(8+\eta\left(2\pi\max|\lambda_j|\frac{1}{\eta}\right)^2+1\right)\\
					&= \left(9+\eta\left(2\pi\max|\lambda_j|\frac{1}{\eta}\right)^2\right)\frac{1}{\rho^2}.
\end{split}
\]

\end{proof}
\begin{lemma}[Total Number of Bits Used]\label{resources_used}
	Given a mid-rise quantizer $\mathscr{A}=\mathscr{A}_0+\imath\mathscr{A}_0$ with length $2L$ and $r\in\NN$, if $q\in\mathscr{A}^m$ is a quantized sample from the alphabet, then  $D_\rho S_\rho^r q\in\CC^\eta$ can be encoded by $\eta\cdot2r\log(2Lm)$ bits.
\end{lemma}

\begin{proof}
	Given the assumption above, each entry $q_j$ of $q$ is a number of the form 
\[
q_j=\big((2s_j+1)+\imath(2t_j+1)\big)\frac{\delta}{2},\quad-L\leq s_j,\, t_j\leq L-1.
\] 
	Then, each entry in $S_\rho q$ is the average of $\rho$ entries in $q$, which has the form 
\[
	(S_\rho q)_j=\big((2\tilde{s}_j+\rho)+\imath(2\tilde{t}_j+\rho)\big)\frac{\delta}{2\rho},\quad-Lm\leq \tilde{s}_j,\, \tilde{t}_j\leq (L-1)m.
\]	
	There are at most $((2L-1)m+1)^2\leq (2Lm)^2$ choices per entry. Note that there are $(2Lm)^2$ choices instead of $(2L\rho)^2$ as we need to account for the first $\rho-1$ rows, which sums $m-\rho$ terms. Iterating $r$ times, there are $(2Lm)^{2r}$ choices for each entry of $S_\rho^r q$. Thus, the vector $D_\rho S_\rho^r q$ can be encoded by $\mathscr{R}=\eta\cdot2r\log(2Lm)$ bits. 
\end{proof}

\begin{proof} of Theorem \ref{thm:high_order}:

	To estimate the reconstruction error, we note that
\[
	D_\rho S_\rho^2\Phi_{m,k}=\Phi_{m/\rho,k}\bar{C}_{m,\rho}^2,
\]
	which follows from Lemma \ref{twisting_uni_high}. Moreover, $(D_\rho S_\rho^2 \Phi_{m,k})^\dagger=\bar{C}_{m,\rho}^{-2}\Scal^{-1} \Phi_{m/\rho,k}^\ast$, where $\Scal=\Phi_{m/\rho,k}^\ast\Phi_{m/\rho,k}$ has lower frame bound $\frac{m}{\rho} C_{\phi_0}$. Since $\|\bar{C}_{m,\rho}^{-1}\|_2\leq\frac{\pi}{2}$, the reconstruction error is
\begin{equation}
\label{error_2nd_order}
\begin{split}
	\mathscr{E}_{m,\rho}&=\|x-\bar{C}_{m,\rho}\Scal^{-1}\Phi_{m/\rho,k}^\ast q\|_2\\
					&=\|\bar{C}_{m,\rho}^{-2}\Scal^{-1}\Phi_{m/\rho,k}^\ast D_\rho S_\rho^2(\Delta^{(m)})^2u\|_2\\
					&\leq\frac{1}{\rho^2}\|\bar{C}^{-1}\|^2_2\|\Scal^{-1}\|_2\|\Phi_{m/\rho}^\ast D_\rho S_\rho^2\Delta^2 u\|_2\\
					&\leq\frac{\pi^2}{4\eta C_{\phi_0}}\bigg(9+\eta(2\pi\max_{1\leq j\leq k}|\lambda_j|\frac{1}{\eta})^2\bigg)\|u\|_\infty\frac{1}{\rho^2},
\end{split}
\end{equation}
	where $\{v_j\}_j\subset\CC^m$ denotes the canonical basis in $\CC^m$, the first inequality comes from Proposition \ref{fra_var_est_high}, and the second follows from Lemma \ref{uni_var_higher}. Here, we see that the error decays quadratically with respect to the oversampling rate $\rho$.
	
	As for the bits used, note that $\frac{1}{m}= \frac{1}{\eta}\cdot\frac{1}{\rho}$ and
\[
	e^{-\frac{1}{2\eta}\mathscr{R}}=\frac{1}{(2Lm)^2}=\frac{1}{(2L\eta)^2}\frac{1}{\rho^2},
\]
	where $\mathscr{R}=\eta\cdot4\log(2Lm)$ comes from Lemma \ref{resources_used}. Thus, we have 
\begin{equation}
	\mathscr{E}(\mathscr{R})\leq\frac{\pi^2}{4\eta C_{\phi_0}}\bigg(9+\eta(2\pi\max_{1\leq j\leq k}|\lambda_j|\frac{1}{\eta})^2\bigg)\|u\|_\infty\frac{1}{\rho^2}\leq C_{k,\phi_0,L,\eta}\|u\|_\infty 2^{-\frac{1}{2\eta}\mathscr{R}},
\end{equation}
	where $C_{k,\phi_0,L,\eta}\leq\frac{\pi^2}{4\eta C_{\phi_0}}\bigg(9+\eta(2\pi\max_{1\leq j\leq k}|\lambda_j|\frac{1}{\eta})^2\bigg)(2 L\eta)^2$, independent of $\rho$.

\end{proof}

	Lemma \ref{non-commute} shows that $\Delta^{-1}$ and $\bar{\Delta}_\rho$ do not commute, and such non-commutativity limits the potential to generalize alternative decimation to higher orders. For the sake of demonstration, we show explicit calculation in Appendix \ref{higher_order_diff} which highlights the difficulty in the generalization of our results. Thus, to achieve exponential error decay with respect to the bit usage for higher order $\Sigma\Delta$ quantization schemes, we need to employ different approaches. The new scheme we propose will be published in a subsequent manuscript.

\section{Acknowledgement}
The author greatly acknowledges the support from ARO Grant W911NF-17-1-0014, and John Benedetto for the thoughtful advice and insights. Further, the author appreciates the constructive analysis and suggestions of the referees.


\begin{appendix}

\section{Limitation of Alternative Decimation: Third Order Decimation}\label{higher_order_diff}
	The non-commutativity between $\bar{\Delta}_\rho$ and $\Delta^{-1}$ results in incomplete difference scaling when applying $D_\rho S_\rho^r$ on $\Delta^r$, creating substantial error terms. This phenomenon already occurs for $r=3$.
	
\begin{proposition}\label{third_order_fail}
	Given $m,\rho\in\NN$ with $\rho\mid m$, the third order decimation satisfies $D_\rho S_\rho^3\Delta^3=\frac{1}{\rho^3}(\Delta^{(\eta)})^3 D_\rho+O(\rho^{-2})$. In particular, $D_\rho S_\rho^3$ only yields quadratic error decay with respect to the oversampling ratio $\rho$.
\end{proposition}
	
	First, by noting that $\Delta^{-1}\bar{\Delta}_\rho\Delta=\Ecal$ as in Lemma \ref{non-commute}, one has
	
\begin{equation}
\label{3rd-order-deci}
\begin{split}
	D_\rho S_\rho^3\Delta^3&=\frac{1}{\rho^3}D_\rho\bar{\Delta}_\rho\Delta^{-1}\bar{\Delta}_\rho\Delta^{-1}\bar{\Delta}_\rho\Delta^2\\
				&=\frac{1}{\rho^3}D_\rho\bar{\Delta}_\rho(\Delta^{-1}\bar{\Delta}_\rho\Delta)\Delta^{-2}\bar{\Delta}_\rho\Delta^2\\
				&=\frac{1}{\rho^3}D_\rho\bar{\Delta}_\rho(\bar{\Delta}_\rho+\Ecal)\Delta^{-1}(\bar{\Delta}_\rho+\Ecal)\Delta\\
				&=\frac{1}{\rho^3}D_\rho\bar{\Delta}_\rho(\bar{\Delta}_\rho+\Ecal)(\bar{\Delta}_\rho+\Ecal+\Delta^{-1}\Ecal\Delta)\\
				&=\frac{1}{\rho^3}D_\rho\bigg(\bar{\Delta}_\rho^3+\bar{\Delta}_\rho^2\Ecal+\bar{\Delta}_\rho^2(\Delta^{-1}\Ecal\Delta)+\bar{\Delta}_\rho\Ecal\bar{\Delta}_\rho+\bar{\Delta}_\rho\Ecal^2+\bar{\Delta}_\rho\Ecal(\Delta^{-1}\Ecal\Delta)\bigg).
\end{split}
\end{equation}

	We shall calculate all terms one-by-one.
\begin{lemma}
\label{third_order_calc}
	We have the following equalities:
\begin{itemize}

\item[(1)]:\[	(D_\rho\bar{\Delta}_\rho^2\Ecal)_{l,s}=\delta(s-(m-\rho))\bigg(\delta(l-1)-\delta(l-2)\bigg),\]
\item[(2)]:\[	(D_\rho\bar{\Delta}_\rho^2(\Delta^{-1}\Ecal\Delta))_{l,s}=\left\{\begin{array}{lcl}
							-\rho&\text{if}& (l,s)=(1,m-\rho-1)\\
							\rho&\text{if}& (l,s)=(1,m-\rho)\\
							0&\text{otherwise}\end{array}\right.,\]
\item[(3)]:\[	(D_\rho\bar{\Delta}_\rho\Ecal\bar{\Delta}_\rho)_{l,s}=\delta(l-1)\bigg(\delta(s-(m-\rho))-\delta(s-(m-2\rho))\bigg),\]
\item[(4)]:\[	(D_\rho\bar{\Delta}_\rho\Ecal^2)_{l,s}=\delta(l-1)\delta(s-(m-\rho)),\]
\item[(5)]:\[	(D_\rho\bar{\Delta}_\rho\Ecal(\Delta^{-1}\Ecal\Delta))_{l,s}=(m-\rho)\delta(l-1)\bigg(\delta(s-(m-\rho))-\delta(s-(m-\rho-1)\bigg),\]
\end{itemize}
	where given $n\in\NN$, $[n]:=\{1,\dots, n\}$.
	In particular, $D_\rho\big(\bar{\Delta}_\rho^2(\Delta^{-1}\Ecal\Delta)+\bar{\Delta}_\rho\Ecal(\Delta^{-1}\Ecal\Delta)\big)=O(m)$, and $D_\rho(\bar{\Delta}_\rho^2\Ecal+\bar{\Delta}_\rho\Ecal\bar{\Delta}_\rho+\bar{\Delta}_\rho\Ecal^2)=O(1)$.
\end{lemma}

\begin{proof}
	We will first compute each term without the effect of $D_\rho$ since $D_\rho$ is the sub-sampling matrix retaining only the $t\rho$-th rows for $t\in[\eta]$.

\begin{itemize}
\item[(1), (3)] First, note that $(\bar{\Delta}_\rho\Ecal)_{l,s}=\delta(l-\rho)\delta(s+\rho)$, so 
\[
(\bar{\Delta}_\rho^2\Ecal)_{l,s}=\delta(s+\rho)(\bar{\Delta}_\rho)_{l,\rho}=\delta(s-(m-\rho))(\delta(l-\rho)-\delta(l-2\rho)). 
\]

	Similarly, 
\[
(\bar{\Delta}_\rho\Ecal\bar{\Delta}_\rho)_{l,s}=\delta(l-\rho)(\bar{\Delta}_\rho)_{m-\rho,s}=\delta(l-\rho)(\delta(s-(m-\rho))-\delta(s-(m-2\rho))).
\]

\item[(5)]	Now, to compute $\Delta^{-1}\Ecal\Delta$, we see that, for $s\neq m$,
\[
	(\Delta^{-1}\Ecal\Delta)_{l,s}=\sum_{j=1}^l(\Ecal_{j,s}-\Ecal_{j,s+1})=l(\delta(m-\rho-s)-\delta(m-\rho-(s+1))),
\]
	and $(\Delta^{-1}\Ecal\Delta)_{l,m}=0$. In particular,
\[
	\Delta^{-1}\Ecal\Delta=\begin{pmatrix}
			0&\hdots&0&-1&1&0&\hdots&0\\
			\vdots&&\vdots&-2&2&\vdots&&\vdots\\
			\vdots&&\vdots&\vdots&\vdots&\vdots&&\vdots\\
			0&\hdots&0&-m&m&0&\hdots&0
			\end{pmatrix},
\]
	where the nonzero columns occur at the $(m-\rho-1)$ and $(m-\rho)$-th positions.

	For $\bar{\Delta}_\rho\Ecal(\Delta^{-1}\Ecal\Delta)$,
\[
\begin{split}
	(\bar{\Delta}_\rho\Ecal(\Delta^{-1}\Ecal\Delta))_{l,s}&=\delta(l-\rho)(\Delta^{-1}\Ecal\Delta)_{m-\rho,s}\\
					&=\delta(l-\rho)(m-\rho)(\delta(s-(m-\rho))-\delta(s-(m-\rho-1))).
\end{split}
\]

\item[(4)]	Note that $\bar{\Delta}_\rho\Ecal^2=\bar{\Delta}_\rho\Ecal$. The result then follows from the calculation on the first term.

\item[(2)]	Finally, as $\Delta^{-1}\Ecal\Delta$ only has non-zero entries on the $(m-\rho-1)$ and $(m-\rho)$-th columns, and the two columns differ by a sign, it suffices to calculate the $(m-\rho)$-th column of $\bar{\Delta}_\rho^2(\Delta^{-1}\Ecal\Delta)$.
\[
\begin{split}
	(\bar{\Delta}_\rho(\Delta^{-1}\Ecal\Delta))_{l,m-\rho}&=\sum_{j=1}^mj(\bar{\Delta}_\rho)_{l,j}\\
					&=\left\{\begin{array}{lcl}
						l-(l-\rho)=\rho&\text{if}&l>\rho\\
						l-(l-\rho+m)=-(m-\rho)&\text{if}&l<\rho\\
						l=\rho&\text{if}&l=\rho.
						\end{array}\right.
\end{split}
\]
	Then,
\[
\begin{split}
	(\bar{\Delta}_\rho^2(\Delta^{-1}\Ecal\Delta))_{l,m-\rho}&=\sum_{j=1}^m(\bar{\Delta}_\rho)_{l,j}(\bar{\Delta}_\rho(\Delta^{-1}\Ecal\Delta))_{j,m-\rho}\\
				&=\left\{\begin{array}{lcl}
					-m&\text{if}& l\in[2\rho-1]\backslash\{\rho\}\\
					\rho&\text{if}& l=\rho\\
					0&\text{otherwise}.
					\end{array}\right.
\end{split}
\]
\end{itemize}
\end{proof}

\begin{proof} of Proposition \ref{third_order_fail}:

	From \eqref{3rd-order-deci} and Lemma \ref{third_order_calc}, we see that
\[
	D_\rho S_\rho^3\Delta^3=\frac{1}{\rho^3}D_\rho\bar{\Delta}_\rho^3+\frac{\eta}{\rho^2}\Ecal_1+\frac{1}{\rho^3}\Ecal_2=\frac{1}{\rho^3}D_\rho\bar{\Delta}_\rho^3+O(\rho^{-2}),
\]
	where
\[
	(\Ecal_1)_{l,s}=\frac{1}{m}\bigg(D_\rho(\bar{\Delta}_\rho^2(\Delta^{-1}\Ecal\Delta)+\bar{\Delta}_\rho\Ecal(\Delta^{-1}\Ecal\Delta))\bigg)_{l,s}=\left\{\begin{array}{lcl}
			-1&\text{if}& (l,s)=(1,m-\rho-1)\\
			1&\text{if}& (l,s)=(1,m-\rho)\\
			0&\text{otherwise,}
			\end{array}\right.
\]
	and
\[
	(\Ecal_2)_{l,s}=\bigg(D_\rho(\bar{\Delta}_\rho^2\Ecal+\bar{\Delta}_\rho\Ecal\bar{\Delta}_\rho+\bar{\Delta}_\rho\Ecal^2)\bigg)_{l,s}=\left\{\begin{array}{lcl}
			-1&\text{if} &(l,s)=(2,m-\rho)\,\text{ or }\, (1,m-2\rho)\\
			3&\text{if} &(l,s)=(1,m-\rho)\\
			0&\text{otherwise}.
			\end{array}\right.
\]
\end{proof}

	Even in higher order cases, alternative decimation still only yields quadratic error decay with respect to the oversampling ratio, as can be seen in Figure \ref{fig:compare_4} and \ref{fig:compare_5}. 
	
	Alternative decimation is limited by this incomplete cancellation, but canonical decimation has even worse error decay. Contrary to the quadratic decay for alternative decimation, canonical decimation only has linear decay for high order $\Sigma\Delta$ quantization. The same thing applies to plain $\Sigma\Delta$ quantization, as can be seen in Figure \ref{fig:compare_2}.

\section{Numerical Experiments}\label{num_exp}
	Here, we present numerical evidence that the alternative decimation on frames has linear and quadratic error decay rate for the first and the second order, respectively. Moreover, it is shown that the canonical decimation, as described in Remark \ref{tilde_S}, is not suitable for our purpose when $r\geq2$.

	Recall that given $m,r,\rho$, one can define the canonical decimation operator $D_\rho\tilde{S}_\rho^r\in\RR^{\eta\times m}$, where $\tilde{S}_\rho\in\RR^{m\times m}$ is a circulant matrix.

\subsection{Setting}
	In our experiment, we look at three different quantization schemes: alternative decimation, canonical decimation, and plain $\Sigma\Delta$. Given observed data $y\in\CC^m$ from a frame $E\in\CC^{m\times k}$ and $r\in\NN$, one can determine the quantized samples $q\in\CC^m$ by
\[
	y-q=\Delta^r u
\]
	for some bounded $u$. The three schemes differ in the choice of dual frames:
\begin{itemize}
\item	Alternative decimation: $\tilde{x}=(D_\rho S_\rho^r E)^\dagger D_\rho S_\rho^r q=F_a q$.
\item	Canonical decimation:  $\tilde{x}=(D_\rho \tilde{S}_\rho^r E)^\dagger D_\rho \tilde{S}_\rho^r q=F_c q$.
\item	Plain $\Sigma\Delta$: $\tilde{x}=E^\dagger q=F_p q$.
\end{itemize}
	
	For each experiment, we use the mid-rise quantizer $\mathscr{A}$ and fix $k=55, \delta=0.5, L=100$, and $\eta=65$. For each $\rho$, we set $m=\rho\eta$ and pick 10 randomly generated vectors $\{x^j\}_{j=1}^{10}\subset\CC^k$. $\Sigma\Delta$ quantization on each signal gives $\{q^j\}_{j=1}^{10}\subset\CC^m$. The maximum reconstruction error over the 10 experiments is recorded, namely
\[
	\mathscr{E}_{i}=\max_{1\leq j\leq10}\|x^j-F_i q^j\|_2,\quad i\in\{a,c,p\}.
\]
	
	The frame in our experiment is
\[
(E^{m,k})_{l,j}=(E)_{l,j}=\frac{1}{\sqrt{k}}(\exp(-2\pi\imath(l+1)(j+1)/m))_{l,j}.
\]

\begin{figure}
	\centering
	\begin{subfigure}[t]{0.4\textwidth}
		\includegraphics[width=\textwidth]{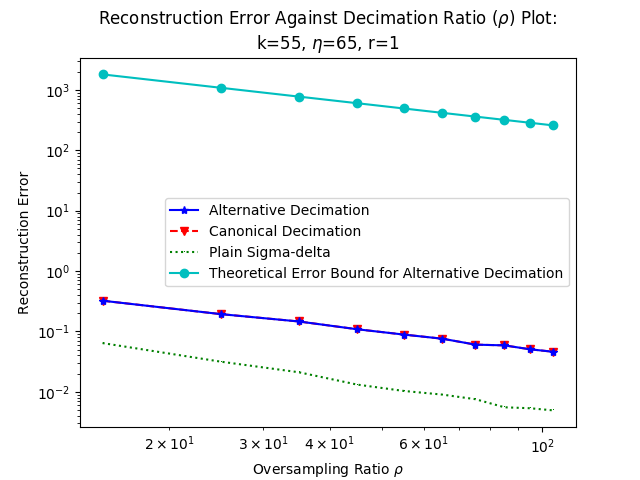}
		\caption{$r=1$.}
		\label{fig:compare_1}
	\end{subfigure}
	~
	\begin{subfigure}[t]{0.4\textwidth}
		\includegraphics[width=\textwidth]{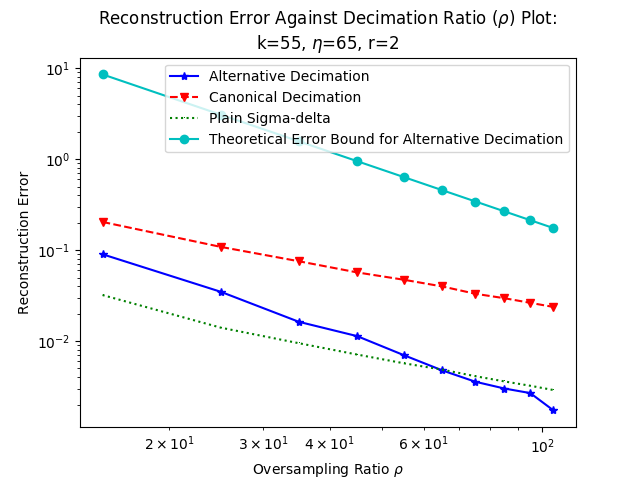}
		\caption{$r=2$.}
		\label{fig:compare_2}
	\end{subfigure}
	~
	\begin{subfigure}[t]{0.3\textwidth}
		\includegraphics[width=\textwidth]{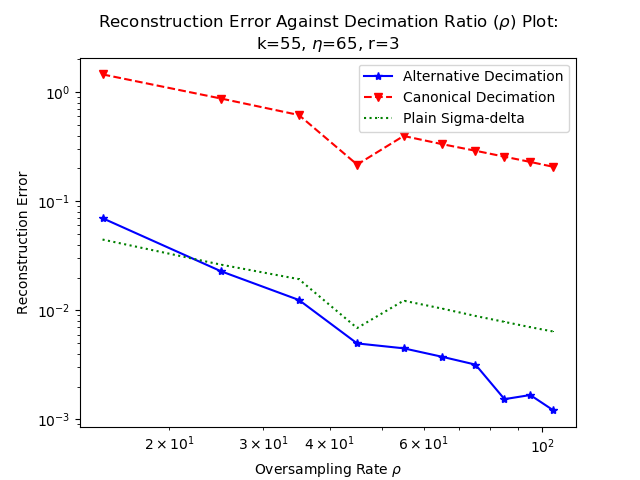}
		\caption{$r=3$.}
		\label{fig:compare_3}
	\end{subfigure}
	~
	\begin{subfigure}[t]{0.3\textwidth}
		\includegraphics[width=\textwidth]{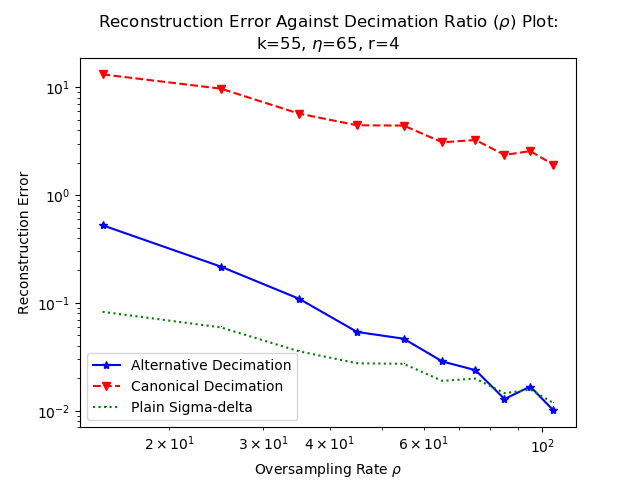}
		\caption{$r=4$.}
		\label{fig:compare_4}
	\end{subfigure}
	~
	\begin{subfigure}[t]{0.3\textwidth}
		\includegraphics[width=\textwidth]{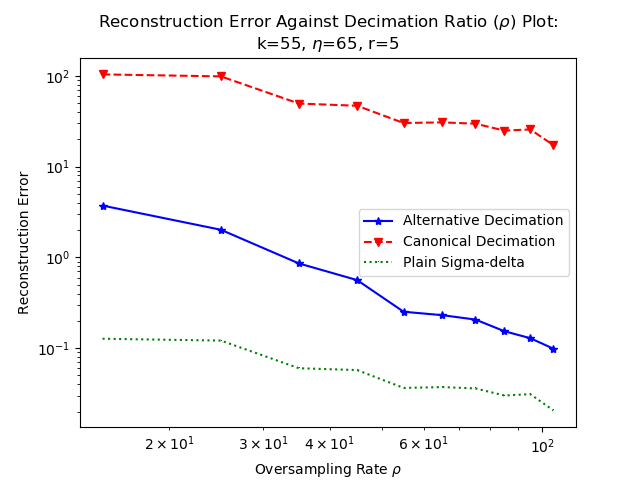}
		\caption{$r=5$.}
		\label{fig:compare_5}
	\end{subfigure}
	\caption{The log-log plot for reconstruction error against the decimation ratio $\rho$ for different quantization schemes. In the case $r=1$, alternative decimation coincides with canonical decimation. For $r\geq2$, alternative decimation has better error decay rate than both canonical decimation and plain $\Sigma\Delta$ quantization.}
	\label{fig:recon_compare}
\end{figure}

	First, we shall compare alternative decimation with plain $\Sigma\Delta$ quantization from Figure \ref{fig:recon_compare}. For $r=1$, alternative decimation performs worse than plain $\Sigma\Delta$ quantization, as plain $\Sigma\Delta$ quantization benefits from the smoothness of the frame elements, having decay rate $O((\frac{m}{k})^{-5/4})$ proven in \cite{JB_AP_OY_2006}. However, for $r\geq2$, alternative decimation supersedes plain $\Sigma\Delta$ quantization as the better scheme. This can be explained by the boundary effect in finite-dimensional spaces that results in incomplete cancellation for backward difference matrices. We are interested in the case $r=1$ or $2$. As we can see, the theoretical error bound does not have a tight constant, although the decay rate is consistent with our experimental result.

\subsection{Necessity of Alternative Decimation}\label{alt_dec_nec}
	
	The main difference between the alternative decimation operator $D_\rho S_\rho^r$ and the canonical one $D_\rho\tilde{S}_\rho^r$ lies in the scaling effect on difference structures. We have $\tilde{S}_\rho^r=(S_\rho+L)^r$ with $\rho L$ having unit entries on the first $\rho-1$ rows and $0$ everywhere else.
	
	In Figure \ref{fig:recon_compare}, we can see the performance drop-off when switching from alternative decimation to canonical decimation for $r\geq2$. we can see that canonical decimation incurs much worse reconstruction error than the alternative one, while generally having worse decay rate. For demonstration, we show explicitly the difference between alternative and canonical decimation schemes for $r=2$:
\[
\begin{split}
	\tilde{S}_\rho^2\Delta^2&=(S_\rho+L)^2\Delta^2\\
					&=S_\rho^2\Delta^2+(LS_\rho+S_\rho L+L^2)\Delta^2\\
					&=S_\rho^2\Delta^2+L(S_\rho+L^2)\Delta^2+S_\rho L\Delta^2.
\end{split}
\]
	Since $D_\rho L=0$, we are left with $D_\rho S_\rho L\Delta^2$. Now,
\[
	(L\Delta^2)_{l,j}=\left\{\begin{array}{ll}
			\frac{-1}{\rho}&\text{if}\quad1\leq l\leq \rho-1,\, j=m-1,\\
			\frac{1}{\rho}&\text{if}\quad1\leq l\leq \rho-1,\, j=m,\\
			0&\text{otherwise}.
			\end{array}\right.
\]
	Then, we see that
\[
	(DS_\rho L\Delta^2)_{l,j}=\left\{\begin{array}{ll}
			\frac{-(\rho-1)}{\rho^2}&\text{if}\quad l=1,\, j=m-1,\\
			\frac{\rho-1}{\rho^2}&\text{if}\quad l=1,\, j=m,\\
			0&\text{otherwise}.
			\end{array}\right.
\]
	We see that $D_\rho\tilde{S}_\rho^2\Delta^2=O(\rho^{-1})$, hence the linear decay for $r=2$.

\end{appendix}

\bibliographystyle{amsplain}

\bibliography{ref}

\end{document}